\documentclass{amsart}
\usepackage[final]{hao}

\title[Continuum Limit for One-Dimensional Moir\'e Relaxation]{Continuum Limit of a One-Dimensional Atomistic Model for Moir\'e Relaxation}
\author{Hao Lan, Michael Hott, Jeff Calder, Alexander B. Watson}
\date{\today}

\makeatletter
\def\@settitle{%
  \begin{center}%
    \baselineskip14\p@\relax
    \bfseries
    \uppercasenonmath\@title
    \@title\par
    \ifx\@empty\@date\else
      \vskip 0.75em
      {\normalfont\normalsize \@date\par}%
    \fi
  \end{center}%
}
\def\@adminfootnotes{%
  \let\@makefnmark\relax  \let\@thefnmark\relax
  \ifx\@empty\@subjclass\else \@footnotetext{\@setsubjclass}\fi
  \ifx\@empty\@keywords\else \@footnotetext{\@setkeywords}\fi
  \ifx\@empty\thankses\else \@footnotetext{%
    \def\par{\let\par\@par}\@setthanks}\fi
}
\makeatother

\begin{document}

\begin{abstract}
  We establish convergence of minimizers for a minimal 
  one dimensional atomic-scale model for relaxation in moir\'e materials to those of a one-dimensional continuum limit functional analogous to those commonly used in practice to model relaxation. 
The atomic-scale model consists of two periodic chains containing $N$ and $N+1$ atoms per common cell and is exactly moir\'e-periodic at the atomic scale. Its energy combines harmonic nearest-neighbor intralayer interactions with nonlocal interlayer interactions generated by an even $C^2$ pair potential with polynomial decay. The associated continuum energy couples linear elasticity to a periodic misfit potential, obtained by periodizing the pair potential, which penalizes unfavorable local stackings. 

\end{abstract} 

\maketitle

\section{Introduction}
Moir\'e materials are formed by stacking two-dimensional crystalline materials with a small relative twist angle or lattice mismatch. The superposition of the lattice periodicities produces a periodic moir\'e pattern whose wavelength is much larger than the lattice constants of either layer. 
Moir\'e materials have emerged as ideal experimental platforms for investigating strongly correlated electronic phases. Notable examples include superconductivity in magic-angle twisted bilayer graphene \cite{Cao2018Superconductivity} and the fractional quantum anomalous Hall effect in twisted bilayer $\mathrm{MoTe}_2$ \cite{Park2023FQAH}.

Mechanical relaxation in bilayer moir\'e materials is the process by which atoms move from their monolayer equilibrium positions in order to minimize their total energy. The relaxed structure is determined by the competition between the elastic cost of deforming each layer and the dependence of the interlayer energy on local stacking. 
This competition tends to enlarge regions of energetically favorable stacking and, when the twist angle is small, localize the transitions between them to narrow domain walls \cite{CarrEtAl2018Relaxation,CazeauxEtAl2023DomainWalls,NamKoshino2017Relaxation}. Such features have been observed experimentally \cite{Yoo2019Reconstruction}. Together with calculations showing that relaxation can substantially modify the low-energy electronic band structure \cite{NamKoshino2017Relaxation}, these observations indicate that accurate prediction of moir\'e materials' electronic properties requires reliable modeling of the relaxed atomic structure.

The standard approach to determining relaxed structures has been to minimize functionals for continuum moir\'e-periodic displacements \cite{CarrEtAl2018Relaxation,CazeauxEtAl2023DomainWalls,CazeauxLuskinMassatt2020,Dai2016MoireTwist,NamKoshino2017Relaxation}.
These functionals are the sum of an intralayer elasticity energy and an interlayer misfit energy which penalizes unfavorable stacking configurations. Both phenomenological \cite{NamKoshino2017Relaxation} and DFT-parametrized misfit energies have been proposed \cite{CarrEtAl2018Relaxation}. Despite the widespread use of such continuum functionals, their rigorous justification from an underlying atomic-scale relaxation model has remained open. 
Providing a clear link between atomic-scale models and their continuum approximations recently became more important with the advent of machine-learned interatomic potentials for moir\'e materials \cite{Wang_2025}. 

The present work provides such a justification for a minimal
one-dimensional model. 
Specifically, we consider two periodic chains with common period $2\pi$, containing $N$ and $N+1$ atoms per period, respectively. The system is, thus, exactly moir\'e-periodic at the atomic scale. 
After rescaling, the atomistic energy consists of harmonic nearest-neighbor intralayer terms and an interlayer interaction generated by an even, sufficiently regular pair potential $V$ with polynomial decay. In the continuum limit, the interlayer interactions generate a periodic misfit potential by summing $V$ over lattice translates, while the intralayer interactions produce the continuum linear elastic energy.
The resulting continuum relaxation functional can be understood as the one-dimensional analogue of the continuum models used in \cite{CarrEtAl2018Relaxation,NamKoshino2017Relaxation}. Our main result, Theorem~\ref{thm:discrete_to_continuum_convergence}, shows that, under the smoothness and decay assumptions on $V$ and a uniform smallness condition on the interlayer interaction, a symmetry-normalized atomistic minimizer $u_d$ converges to the sampled normalized continuum minimizer $u_c$ in the discrete $H_h^2$ norm at rate $O(N^{-1})$.

The limit considered here is the same considered formally in the recent work \cite{jingzhi2025formaljustificationcontinuumrelaxation}, which uses physical values of parameters to argue this limit is realized in twisted bilayer graphene near to the magic angle $\approx 1^{\circ}$. In the one-dimensional setting considered here, we complement our quantitative convergence estimates with a $\Gamma$-convergence result connecting the atomistic and continuum energy functionals \cite{10.1093/acprof:oso/9780198507840.001.0001,BraidesGelli2006}. We leave the extension of these results to two-dimensional atomistic models of layered graphene sheets to future work.

Discrete-to-continuum limits for stacked layers with one layer fixed were introduced in \cite{PhysRevE.96.033003,Espanol}, and their $\Gamma$-convergence established in \cite{golovaty2025hierarchyscalesmodelingweakly}. The works \cite{Cazeaux2017,2016Cazeaux} analyzed relaxation of incommensurate coupled chains, mapping the system to a Frenkel-Kontorova model and deriving a Cauchy-Born strain energy density respectively. The work \cite{CazeauxLuskinMassatt2020} introduced a general framework for relaxing two-dimensional $n$-layer incommensurate stacks, but without proving convergence between atomic-scale and continuum functionals coupling elasticity to a local stacking energy. A novel approach to relaxation of moir\'e materials where the mechanical and electronic degrees of freedom are coupled was recently proposed in \cite{tu2026relaxationincommensuratestructuresquantum}. 

The proof of Theorem~\ref{thm:discrete_to_continuum_convergence} combines consistency of the Euler--Lagrange equations with discrete stability. Section~\ref{sec:cont_min} characterizes the continuum minimizer and establishes the regularity and symmetry properties needed for comparison with the atomistic model. Section~\ref{sec:discrete_min} establishes existence of atomistic minimizers and a discrete stability estimate. In Section~\ref{sec:consistency_convergence}, we show that the sampled continuum minimizer approximately satisfies the discrete Euler--Lagrange equations. Combining this consistency estimate with discrete stability and symmetry normalization yields our main result, Theorem~\ref{thm:discrete_to_continuum_convergence}. We then use $\Gamma$-convergence in Section~\ref{sec:gamma} to show that minimizers converge without the smallness condition. The numerical examples in Section~\ref{sec:numerical} validate first-order convergence of the computed atomistic displacement and provide evidence for second-order convergence of the reconstructed relative displacement. The appendices contain the derivation of the simplified atomistic model and the auxiliary results supporting the analysis.

\subsection*{Acknowledgements}
AW acknowledges support from National Science Foundation Grant DMS-2406981. JC acknowledges support from an Albert and Dorothy Marden Professorship, a Simons Fellowship, and National Science Foundation Grant DMS-2436333. MH acknowledges support through the Simons Targeted Grant on Moir\'e Materials Magic. AW acknowledges stimulating conversations with Mitchell Luskin, Daniel Massatt, Paul Cazeaux, Dumitru Calugaru, Dmitry Golovaty, and Raghav Venkatraman. The authors used AI-assisted tools to help develop the numerical code and prepare the figures. The authors take full responsibility for the content of this paper.


\section{Well-posedness of the continuum energy and properties of minimizers} \label{sec:cont_min}

We consider the following continuum relaxation functional
\begin{equation}\label{eq:E12}
E\left[u_1,u_2\right] = \inty{-\pi}{\pi}{\tfrac12 u_1'(x)^2+\tfrac12 u_2'(x)^2
+ \Phi\!\left(x+u_1(x)-u_2(x)\right)}{x},
\end{equation}
where $u_1, u_2 \in \Hp:=\left\{u \in H_{\mathrm{loc}}^1(\R): u(x+2 \pi)=u(x) \right\}$ represent the displacements of the two layers. Without loss of generality, we identify functions $u \in \Hp$ with their restriction to the fundamental cell $[-\pi,\pi]$ with periodic boundary conditions $u(\pi) = u(-\pi)$.
We can impose periodicity conditions pointwise since $H^1(-\pi, \pi) \hookrightarrow C([-\pi, \pi])$.

The potential $\Phi \in C^{1,1}(\R)$ which describes the interlayer misfit, is assumed to be even and $2\pi$--periodic. A canonical continuum example is the cosine potential $\Phi(x) = \cos(x)$. 

It is convenient to introduce the variables
\begin{equation}\label{eq:uv_change}
u_+(x):=u_1(x)+u_2(x),
\qquad
v(x):=x+u_1(x)-u_2(x).
\end{equation}
Equivalently,
\[
u_1(x)=\tfrac12\bigl(u_+(x)+v(x)-x\bigr),
\qquad
u_2(x)=\tfrac12\bigl(u_+(x)-v(x)+x\bigr).
\]
A direct computation gives
\begin{equation}\label{eq:E_uv}
E\left[u_1,u_2\right]
=
\inty{-\pi}{\pi}{\frac14\,u_+'(x)^2}{x} + \inty{-\pi}{\pi}{\frac14\,\bigl(v'(x)-1\bigr)^2+\Phi\!\bigl(v(x)\bigr)}{x} .
\end{equation}
Since $u_1,u_2\in \Hp$, we have $u_+\in \Hp$ and $v \in x + \Hp$ defined by
\begin{equation}\label{eq:Hv}
x + \Hp = \left\{x + u(x) : u \in \Hp\right\}.
\end{equation}
Note that we have $v(x+2\pi) = v(x)+2\pi$ for $v\in x + \Hp$. Again, we can equivalently consider functions $v \in \H{1}$ with the boundary condition $v(\pi)-v(-\pi)=2\pi$. The first term of \eqref{eq:E_uv} is minimized by any constant $u_+ = C$. Thus, the minimization problem for $E\left[u_1,u_2\right]$ reduces entirely to the $v$ part. Define $W := 2\Phi$, we have that 
$$
E\left[u_1,u_2\right]=\frac14\inty{-\pi}{\pi}{u_+'(x)^2}{x}+\frac12J[v]-\frac{\pi}{2},
$$ 
where
\begin{equation}\label{eq:v}
J[v] := \inty{-\pi}{\pi}{\frac{1}{2}v'(x)^2 + W\!\bigl(v(x)\bigr)}{x}.
\end{equation}
Then minimizing \eqref{eq:E_uv} is equivalent to minimizing \eqref{eq:v} over $v\in x + \Hp$. First observe that $J$ depends only on $v'$ and on $W$, which are both $2\pi$--periodic. Therefore $J$ is invariant under vertical shifts by $2\pi\mathbb{Z}$:
\[
J[v+2\pi k]=J[v]\qquad\text{for all }v\in x+\Hp,\ k\in\mathbb{Z}.
\]
For $d\in\R$ define translation $(\tau_d v)(x):=v(x-d)$. By change of variables we obtain
\[
J[\tau_d v]=J[v].
\]
In particular, if $v_0$ is a global minimizer, then so is $\tau_d v_0+2\pi k$ for every $d\in\R$ and $k\in\mathbb{Z}$. We will prove the following theorem on minimizers of \eqref{eq:v} in $ x + H^1_{\mathrm{per}}$.
\begin{theorem}\label{thm:v}
There exists a unique global minimizer $v_0\in x + \Hp$ of \eqref{eq:v} satisfying $v_0(0)=0$.
All other minimizers have the form $v = \tau_d v_0 + 2\pi k$ for $d\in \R, k\in \Z$.
In other words,
\[
    \operatorname*{argmin}_{ x + \Hp} J \;=\; \big\{\, v_0(x - d) + 2\pi k : d \in \R, \; k \in \mathbb{Z} \,\big\}.
\]
Moreover, $v_0 \in C^{2,1}(\R)$ is odd and satisfies $v_0'>0$.
\end{theorem}
\begin{figure}[h]
  \centering
  \includegraphics[width=0.67\textwidth]{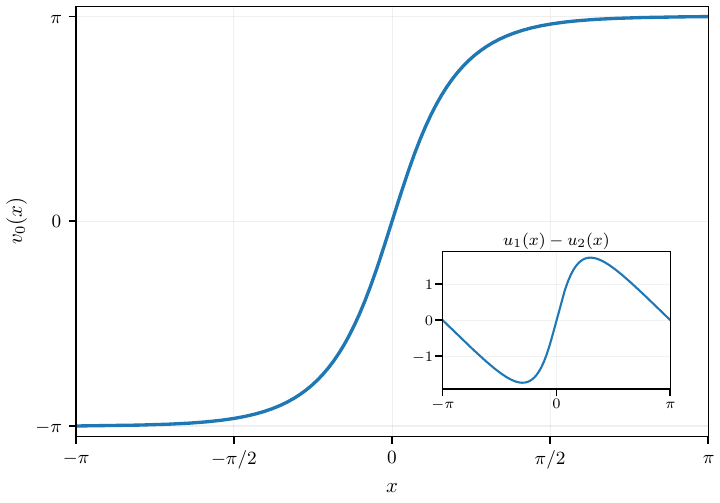}
  \caption{Numerical approximation of the normalized continuum minimizer $v_0$ for the Lennard--Jones interaction described in Section \ref{sec:numerical}, with layer separation $d=2\pi$, $\sigma=1.8\pi$, and $\epsilon=16$. The normalization is $v_0(0)=0$, with $v_0(x+2\pi)=v_0(x)+2\pi$. The inset shows the periodic relative displacement $u_1(x)-u_2(x)=v_0(x)-x$.}
  \label{fig:lj-continuum-minimizer}
\end{figure}
For these parameters, Figure \ref{fig:lj-continuum-minimizer} shows broad plateaus joined by a narrow, smooth transition. Thus the registry varies slowly over much of the period, while most of its change is concentrated in the transition region. With $u_+=0$, the individual layer displacements are recovered from this profile as
$$
 u_1(x)=\frac{v_0(x)-x}{2},\quad
 u_2(x)=-\frac{v_0(x)-x}{2}.
$$

The proof of Theorem \ref{thm:v} requires several ingredients which are collected below. First, the existence and regularity of a minimizer is a standard application of the direct method in the calculus of variations. We establish monotonicity in Proposition \ref{lem:monotone} and uniqueness in Proposition \ref{lem:uniqueness} using the first integral of the Euler--Lagrange equation and the uniqueness for the associated initial value problem. Finally, we establish stability of the Euler--Lagrange equation in $H^2$ norm in Proposition \ref{lem:stability}.
\begin{proof}[Proof of Theorem \ref{thm:v}]
For the existence, let $v_n= x + u_n \in x+H^1_{\mathrm{per}}(-\pi,\pi)$ be a minimizing sequence. Using the symmetry $J[v+2\pi k]=J[v]$, we can arrange it so that the means $(u_n)$ are uniformly bounded. Now applying the Poincar\'e-Wirtinger  
inequality yields a uniform $H^1$ bound on $u_n$. Following the direct method as in \cite[\S3.4.1, Thm. 3.30]{dacorogna_direct}, $\inf_{x+H^1_{\mathrm{per}}}J$ attains its minimum.

For the regularity, we claim that every global minimizer $v$ of $J$ on $x+\Hp$ is a $C^{2,1}$ strong solution of the Euler--Lagrange equation
\begin{equation}\label{eq:EL_v}
 v''(x)=W'\bigl( v(x)\bigr),
\qquad
 v(x+2\pi)= v(x) + 2\pi.
\end{equation}
Indeed, by Evans' Euler--Lagrange theorem for minimizers \cite[Ch. 8, Thm. 4]{evans2010partial} any minimizer $v$ satisfies \eqref{eq:EL_v} in the weak sense. Following the standard ODE theory, we can bootstrap directly to obtain $ v\in C^{2,1}(\R)$ since $W\in C^{1,1}(\R)$.

Now let \(v\in x+\Hp\) be any global minimizer. By the regularity result, $v\in C^{2,1}(\R)$, 
the endpoint values \(v(\pm\pi)\) are well-defined. Choose \(k\in\mathbb Z\) such that
\[
v(-\pi)-2\pi k \le 0 < v(\pi)-2\pi k.
\]
By continuity, there exists \(x_0\in[-\pi,\pi]\) with \(v(x_0)-2\pi k=0\). Define
\[
v_0 := \tau_{-x_0}v - 2\pi k,
\qquad\text{so that}\qquad
v_0(0)=0.
\]
By the symmetry of \(J\), \(v_0\) is again a
global minimizer in \(x+H^1_{\mathrm{per}}\) satisfying \eqref{eq:EL_v}. Now define \(w(x):=-v_0(-x)\). A direct computation gives
\[
w''(x)=-v_0''(-x)
      =-W'(v_0(-x))
      =W'(-v_0(-x))
      =W'(w(x)),
\]
so \(w\) satisfies the same ODE as \(v_0\) on \((-\pi,\pi)\). Moreover,
\[
w(0)=0=v_0(0),\qquad w'(0)=v_0'(0).
\]
The uniqueness for the ODE initial value problem gives \(w = v_0\) on \([-\pi,\pi]\). Therefore \(v_0\) is odd. Combining oddness with
\(v_0(\pi)-v_0(-\pi)=2\pi\), we obtain
\[
v_0(\pi)=\pi,\qquad v_0(-\pi)=-\pi.
\]
Thus every global minimizer is related by symmetries of the functional to an odd global minimizer satisfying $v_0(0)=0$ and $v_0(\pm\pi)=\pm\pi$. Monotonicity and uniqueness of $v_0$ follow from Propositions \ref{lem:monotone} and \ref{lem:uniqueness} below.
\end{proof}
\begin{proposition}[Monotonicity]\label{lem:monotone}
Let $v \in x+\Hp$ minimize \eqref{eq:v}. Then $v'(x) > 0$ on $\R$, and there exists a constant $c > -\min_{\R} W$ such that
\begin{equation}\label{eq:EL_v_first_integral}
\frac12 v'(x)^2 - W(v(x)) = c \qquad \text{for all } x \in \R.
\end{equation}
\end{proposition}

\begin{proof}
Multiplying \eqref{eq:EL_v} by $v'$ and integrating gives \eqref{eq:EL_v_first_integral} for some $c\in\R$. Fix $s_0\in\R$ at which $W$ attains its minimum. Since $v(x+2\pi)=v(x)+2\pi$ and $v$ is continuous, $v$ takes every real value, so $v(x_0)=s_0$ for some $x_0\in\R$. The identity \eqref{eq:EL_v_first_integral} at $x_0$ gives $c\geq-W(s_0)$. If equality held, then $v'(x_0)=0$. Since $W'(s_0)=0$, the constant function $v\equiv s_0$ would solve the same initial value problem. Uniqueness for the ODE would then contradict $v(x+2\pi)=v(x)+2\pi$. Thus $c>-\min_{\R}W$, so $v'$ never vanishes and has a fixed sign. The identity $v(\pi)-v(-\pi)=2\pi$ shows that this sign is positive.
\end{proof}

\begin{proposition}[Uniqueness]\label{lem:uniqueness}
    There exists a unique minimizer $v$ of \eqref{eq:v} in $x+\Hp$ satisfying $v(0)=0$.
\end{proposition}
\begin{proof}
Existence and phase normalization were established above. Let $v_1,v_2$ be two such minimizers. By Proposition \ref{lem:monotone}, there are $c_{\hi}>-\min_{\R}W$ such that $v_{\hi}'=\sqrt{2(W(v_{\hi})+c_{\hi})}$ for ${\hi}=1,2$. For each ${\hi}$, changing variables $s=v_{\hi}(x)$ over one period and using periodicity gives
\[
2\pi
=\int_{v_{\hi}(-\pi)}^{v_{\hi}(\pi)}\frac{\text{d}s}{\sqrt{2(W(s)+c_{\hi})}}
=\int_0^{2\pi}\frac{\text{d}s}{\sqrt{2(W(s)+c_{\hi})}}.
\]
The last integral is strictly decreasing in $c_{\hi}$, hence $c_1=c_2$. Therefore with the normalization $v_1(0)=v_2(0)=0$ and
\[
v_1'(0)=v_2'(0)=\sqrt{2(W(0)+c_1)}.
\]
Uniqueness for the initial value problem \eqref{eq:EL_v} implies $v_1=v_2$.
\end{proof}

In addition to Theorem \ref{thm:v}, we claim the Euler--Lagrange equation \eqref{eq:EL_v} is stable in $H^2$ in the sense of the following Proposition \ref{lem:stability}. The existence, normalization, monotonicity, and uniqueness results above require no shape restriction on $W$. The following $H^2$-stability estimate is proved under the additional sufficient condition $\Lip(W')<1$, which corresponds to the Poincar\'e-Wirtinger constant for zero-mean $2\pi$-periodic functions.
\begin{proposition}[Stability]\label{lem:stability}
Assume that $W \in C^{1,1}(\R)$ with Lipschitz constant $L:=\Lip(W')<1$. Let $v,w\in x+H^2_{\mathrm{per}}$ and $r\in L^2(-\pi,\pi)$ satisfy
\[
v'' = W'(v), \qquad w'' = W'(w) + r \quad \text{a.e.\ on }(-\pi,\pi).
\]
Set $g:=v-w$ and denote the mean by
\[
(g):=\frac1{2\pi}\int_{-\pi}^{\pi} g(x)\,dx.
\]
Then there exists a constant $C=C(L)>0$ such that
\[
\|g\|_{H^2(-\pi,\pi)} \le C\bigl(\|r\|_{L^2(-\pi,\pi)} + |(g)|\bigr).
\]
\end{proposition}

\begin{proof}
Note that $g''=W'(v)-W'(w)-r$ and $g=v-w\in H^2_\mathrm{per}$. The function \(g-(g)\) is periodic
with zero mean and \(g'\) has zero mean by periodicity. Applying Poincaré--Wirtinger inequality to $g-\left( g \right)$ and then to $g'$ on \([-\pi,\pi]\), where the Poincaré constant is $1$, then 
$$
\|g-\left( g \right)\|_{L^2} \leq\left\|g'\right\|_{L^2} \leq\left\|g''\right\|_{L^2},
$$
and
$$
\left\|g''\right\|_{L^2} \leq\left\|W'(v)-W'(w)\right\|_{L^2}+\|r\|_{L^2} \leq L\|g\|_{L^2}+\|r\|_{L^2} .
$$
Using the triangle inequality
$$
\|g\|_{L^2} \leq\|g-\left( g \right)\|_{L^2}+\sqrt{2 \pi}|\left( g \right)| \leq \left\|g''\right\|_{L^2}+\sqrt{2 \pi}|\left( g \right)|.
$$
Combining the above estimates, we get
$$
\|g''\|_{L^2} \leq \frac{L \sqrt{2 \pi}}{1-L}|\left( g \right)|+\frac{1}{1-L}\|r\|_{L^2}
$$
The same bound holds for $\left\|g-(g)\right\|_{L^2}$ and $\left\|g'\right\|_{L^2}$, so
$$
\|g\|_{H^2(-\pi, \pi)} \leq C(L)\left(\|r\|_{L^2}+|\left( g \right)|\right).
$$
\end{proof}

\section{Well-posedness of the discrete energy} \label{sec:discrete_min}

We consider a minimal atomic-scale model for relaxation in moir\'e materials: two one-dimensional chains of atoms with a small lattice mismatch analogous to a two-dimensional twist, where the atoms interact via pair potentials.

We consider the case where the system has an exact periodic cell given by $[0,2\pi)$. Fix an integer $N \geq 2$, we denote the number of atoms in layer ${\hi}$ within this cell by $N_{\hi}$ with $N_1 = N$ and $N_2 = N+1$. The lattice constants are thus $h_{\hi} = \frac{2\pi}{N_{\hi}}$, and the lattice mismatch $\theta$ is defined by
$$
h_2 = (1 - \theta) h_1 \implies \theta = \frac{1}{N+1}.
$$
For \(\xi=(\xi^1,\xi^2)\), let \(\xi_n^{\hi}\) denote the displacement of the \(n\)-th atom in layer \({\hi}\) from its reference position \(n h_{\hi}\). We define the atomic-scale energy by
\begin{equation} \label{eq:sum_E}
\mathcal{E}_N[\xi]:=\mathcal{E}^1[\xi]+\mathcal{E}^2[\xi],
\end{equation}
where \(V_{\hi,\hi}\) and \(V_{\hi,\hj}\), \(\hi \in\{1,2\}\), are decaying pair potentials and
\begin{equation}\label{eq:energy_raw}
\begin{split}
&\mathcal{E}^{\hi}[\xi]
:= \\
&h_{\hi} \sum_{n=0}^{N_{\hi}-1}\sum_{m\in\mathbb Z}
\left[
V_{\hi, \hi}\bigl(nh_{\hi}+\xi_n^{\hi}-mh_{\hi}-\xi_m^{\hi}\bigr)
+V_{\hi,\hj}\bigl(nh_{\hi}+\xi_n^{\hi}-mh_{\hj}-\xi_m^{\hj}\bigr)
\right].
\end{split}
\end{equation}
Note that the arguments of the potentials are the relative atomic positions, including both reference separation and displacement, and that the infinite sum over \(m\) accounts for the periodicity of the chains. Throughout, any discrete pair \(u=(u^1,u^2)\) is understood to satisfy the periodicity condition
\[
u_{n+N_{\hi}}^{\hi}=u_n^{\hi} .
\]

Rather than considering \eqref{eq:energy_raw} directly, we further simplify this model as follows. For layer ${\hi}$, define the forward and backward differences
$$
\fdi u_n^{\hi}:=\frac{u_{n+1}^{\hi}-u_n^{\hi}}{h_{\hi}}, \quad \bdi u_n^{\hi}:=\frac{u_n^{\hi}-u_{n-1}^{\hi}}{h_{\hi}},
$$
and the periodic discrete Laplacian
$$
\Delta_{\hi} u_n^{\hi}:=\bdi \fdi u_n^{\hi}=\frac{u_{n+1}^{\hi}-2 u_n^{\hi}+u_{n-1}^{\hi}}{h_{\hi}^2} .
$$Denote $\avsum_{n=0}^{N_{\hi}-1} := h_{\hi} \sum_{n=0}^{N_{\hi}-1}$ and take
$I_{\hi} = \{0, \ldots, N_{\hi}-1\}$, $h = \frac{2h_1h_2}{h_1+h_2}$. Then, after rescaling the displacement functions $\xi^i = \frac{h_i}{2 \pi} u^i$ and some further manipulations (see Appendix \ref{sec:simplified} for a detailed derivation), the atomistic energy \eqref{eq:sum_E}-\eqref{eq:energy_raw} can be simplified to
\[
    E_N[u] := E^1[u] + E^2[u],
\]
where
\begin{equation} \label{eq:energy_discrete}
E^{{\hi}}[u] := \frac{1}{2} \avsum_{n \in I_{\hi}} (\nabla_{\hi}^{+}u^{\hi}_n)^2 + \avsum_{n \in I_{\hi}}\sum_{m \in \mathbb{Z}}V\left(\Anm[u]\right), \quad {\hi} = 1,2,
\end{equation}
\begin{equation}\label{eq:argument}
    A^{\hi}_{n,m}\left[u\right] := (-1)^{{\hj}} h n + \frac{h}{h_{\hj}} u_{n}^{\hi} - \frac{h}{h_{\hi}} u_{n+m}^{\hj} - 2 \pi \frac{h}{h_{\hi}} m , \quad n \in I_{\hi}.
\end{equation}
We equip our system with the layerwise weighted $\ell^2$ norm:
$$
\left\|u^{\hi}\right\|_{\ell^2_{{\hi}}}^2:=h_{\hi} \sum_{n=0}^{N_{\hi}-1}\left|u_n^{\hi}\right|^2,
$$
and the rescaled $\ell^2$ norm:
$$
\|u\|_{\ell_h^2}^2:=\left\|u^1\right\|_{\ell^2_{1}}^2+\left\|u^2\right\|_{\ell^2_{2}}^2 .
$$
For the discrete derivatives, we write
$$
\begin{aligned}
& \left\|\nabla_h u\right\|_{\ell_h^2}^2:=\left\|\nabla_{1}^{+} u^1\right\|_{\ell_1^2}^2+\left\|\nabla_{2}^{+} u^2\right\|_{\ell_2^2}^2, \\
& \left\|\Delta_h u\right\|_{\ell_h^2}^2:=\left\|\Delta_{1} u^1\right\|_{\ell_1^2}^2+\left\|\Delta_{2} u^2\right\|_{\ell_2^2}^2.
\end{aligned}
$$
By periodicity, the forward and backward difference operators have the same layerwise \(\ell_h^2\) norm, so the choice of the forward difference in this definition is immaterial. We define the discrete $H^2$ norm by
$$
\|u\|_{H_h^2}^2 := \|u\|_{\ell_h^2}^2+\left\|\nabla_h u\right\|_{\ell_h^2}^2+\left\|\Delta_h u\right\|_{\ell_h^2}^2 .
$$
Note that the additive quantity used below is an equivalent norm:
\[
\|u\|_{H_h^2} \leq \|u\|_{\ell_h^2} + \|\nabla_hu\|_{\ell_h^2} + \|\Delta_hu\|_{\ell_h^2} \leq
\sqrt{3}\,\|u\|_{H_h^2}.
\]
We also use layerwise discrete maximum norm and the maximum norm for the two-layer system,
$$
\left\|u^i\right\|_{\infty, i}:=\max _{n \in I_i}\left|u_n^i\right|, \quad\|u\|_{\infty}:=\max _{i=1,2}\left\|u^i\right\|_{\infty, i}.
$$
Note that one has the estimate
$$
\|u\|_{\infty} \leq\left\|u^1\right\|_{\infty, 1}+\left\|u^2\right\|_{\infty, 2} \leq 2\|u\|_{\infty} .
$$
For the layerwise means, we use
$$
\left(u^{\hi}\right):=\frac{h_{\hi}}{2 \pi}  \sum_{n=0}^{N_{\hi}-1} u_n^{\hi}.
$$

\subsection{Existence of Minimizers}
From this point on, we assume that the potential $V \in C^2(\R)$ is even and satisfies the following polynomial decay condition: for some $r \geq 2$ and $C_V > 0$,
\begin{equation} \label{eq:potential_decay}
|V(x)|+(1+|x|)\left|V'(x)\right|+(1+|x|)^2\left|V''(x)\right| \leq C_V(1+|x|)^{-r}, \quad x \in \R .
\end{equation} 
Under these assumptions, and in particular using the evenness of \(V\), the discrete energy \eqref{eq:energy_discrete} enjoys the following translation and reflection symmetries.
\begin{lemma}[Discrete symmetries] \label{lem:discrete_symmetry} 
For \(d\in \mathbb Z\) and \(c\in\mathbb R\), define the operator
\begin{equation} \label{eq:symmetry_discrete} 
(\tau_{d,c}u)^1_n = u^1_{n+d}+h_2d+h_2c, \qquad (\tau_{d,c}u)^2_n = u^2_{n+d}+h_1c . 
\end{equation}
Here all indices are understood modulo the corresponding layer period \(N_{\hi}\). 
Then 
\[ E_N[\tau_{d,c}u]=E_N[u], \qquad d\in\mathbb Z,\ c\in\mathbb R. \] 
Consequently, the layer means transform under \(\tau_{d,c}\) by \[ ((\tau_{d,c}u)^1)=(u^1)+h_2d+h_2c, \qquad ((\tau_{d,c}u)^2)=(u^2)+h_1c . \]
Moreover, even-ness of $V$ implies the odd-reflection operator 
\[ (\mathcal Ru)^{\hi}_n := -u^{\hi}_{-n},\qquad {\hi}=1,2, \] 
is also a symmetry: \[ E_N[\mathcal Ru]=E_N[u]. \] 
\end{lemma}
Use Lemma \ref{lem:discrete_symmetry}, we obtain the following corollary that allows us to bound the layer means of a configuration without changing its energy.
\begin{corollary}[Mean-fixing by translation] \label{lem:mean_fixing}
For every discrete pair $(u^1,u^2)\in \R^N\times\R^{N+1}$, there exists a discrete pair $(\tilde u^1,\tilde u^2)$ satisfying $E_N[\tilde u] = E_N[u]$ such that
\[
(\tilde u^1)+(\tilde u^2)=0,
\qquad
|(\tilde u^1)|=|(\tilde u^2)|\le \frac h4.
\]
\end{corollary}
We next record a near--far field estimate for the interaction term, obtained directly from the decay assumption \eqref{eq:potential_decay}.
\begin{lemma}[Near--far field splitting]
\label{lem:nf_split}
For every $n\in I_i$ and $m\in\Z$, $\Anm[u]$ defined by \eqref{eq:argument}. Assume \eqref{eq:potential_decay} holds. Then for any $G=V,V',V''$, we have
\[
    \sum_{m\in\Z}\sup_{n\in I_i}|G(\Anm[u])| \leq 13C_V \left(1+\left\|u^1\right\|_{\infty, 1}+\left\|u^2\right\|_{\infty, 2} \right).
\]
\end{lemma}
We obtain an a priori upper bound for the energy of minimizers by testing \eqref{eq:energy_discrete} with the zero configuration \(u^{\hi}_n = 0\). Indeed, Lemma \ref{lem:nf_split} gives
$$
E^i[0] \leq \avsum_{n \in I_{\hi}} \sum_{m \in \mathbb{Z}}\sup_{n\in I_{\hi}}\left|V\left((-1)^{\hj}h n-2 \pi \frac{h}{h_{\hi}} m\right)\right| \leq 26 \pi C_V.
$$
In what follows we restrict attention to configurations satisfying \(E_N[u] \leq C_0\), where \(C_0:=\sup_{N \geq 2} E_N[0] \leq 52 \pi C_V\).
\begin{theorem}[Existence of discrete minimizer]
\label{thm:discrete_minimizer}
Let
\[
\AN
:=
\Bigl\{(u^1,u^2)\in \R^N\times\R^{N+1} :
(u^1)+(u^2)=0,\, |(u^{\hi})|\le \frac h4
\Bigr\}.
\]
There exists a $u_* \in \AN$ such that
\[
E_N[u_*] =\inf_{\AN} E_N[u] =\inf_{\R^N\times\R^{N+1}} E_N[u].
\]
\end{theorem}
The proof is a finite-dimensional direct-method argument once the translation symmetry has been normalized.
We first use Corollary \ref{lem:mean_fixing} to choose a normalized representative in each symmetry class, and then use the zero configuration together with Lemma \ref{lem:nf_split} to obtain a uniform bound on the normalized sublevel set (see Lemma \ref{lem:uniform_boundedness}). These two ingredients make the relevant sublevel set compact, so the energy attains its minimum there.

We first apply Corollary \ref{lem:mean_fixing} to prove that every normalized state $u$ is uniformly bounded given that $E_N[u] \leq C_0$.
\begin{lemma}[Uniform boundedness for normalized states] \label{lem:uniform_boundedness}
Define the sublevel set
$$
\KN := \Bigl\{u \in \AN: E_N[u] \leq C_0\Bigr\},
$$
where $C_0:=\sup_{N \geq 2} E_N[0] \leq 52 \pi C_V$. Assume $u \in \KN$, then
$$
\left\|u^1\right\|_{\infty, 1}+\left\|u^2\right\|_{\infty, 2} \leq 
\frac{h}{2} + 42\pi \sqrt{C_V} + 416\pi^2 C_V.
$$
Since $h = \frac{4\pi}{2N+1}$, the sum of the layerwise maximum norms can be bounded uniformly in $N$ by relaxing the RHS as
$$
\left\|u^1\right\|_{\infty, 1}+\left\|u^2\right\|_{\infty, 2} \leq 
42\pi \sqrt{C_V} + 416\pi^2 C_V + \frac{2\pi}{3}.
$$
\end{lemma}

\begin{proof}[Proof of Theorem \ref{thm:discrete_minimizer}]
By the translation symmetry \eqref{eq:symmetry_discrete} and Corollary \ref{lem:mean_fixing}, we have
$$
\inf _{\R^N \times \R^{N+1}} E_N=\inf _{\AN} E_N .
$$
Since \(0\in \AN\) and \(E_N[0]\le C_0\), we have \(\inf _{\AN} E_N \le C_0\). If \(\inf _{\AN} E_N < C_0\) then any minimizing sequence in \(\AN\) is eventually contained in \(\KN\), and hence
$$
\inf_{\KN}E_N \le \inf _{\AN} E_N.
$$
Since \(\KN\subset \AN\), the reverse inequality is immediate. If \(\inf _{\AN} E_N=C_0\), then
\[
\inf _{\AN} E_N\le E_N[0]\le C_0=\inf _{\AN} E_N,
\]so \(E_N[0]=\inf _{\AN} E_N\), and \(0\in \KN\) gives the same conclusion. Hence,
$$
\inf _{\AN} E_N=\inf_{\KN} E_N.
$$
The set $\KN$ is nonempty and by Lemma \ref{lem:uniform_boundedness} is bounded. Now for fixed $N$, the layer means are continuous and hence $\AN$ is closed. $E_N$ is continuous, hence $\KN$ is also closed, and thus compact in the finite-dimensional space $\R^N \times \R^{N+1}$. Therefore $E_N$ attains its minimum on $\KN$, and hence also on $\R^N \times \R^{N+1}$.
\end{proof}

\subsection{Euler-Lagrange equation} Now take 
\begin{equation} \label{eq:force_map}
F_n^{\hi}\left(u\right) := 2 \sum_{m \in \Z} V'\left(\Anm\left[u\right]\right), \quad  n \in I_{\hi}, \quad {\hi}=1,2.
\end{equation}
Then the Euler-Lagrange equation of \eqref{eq:energy_discrete} can be written as
\begin{equation} \label{eq:EL_compact}
\Delta_{{\hi}} u^{\hi}=F^{\hi}\left(u\right), \quad {\hi}=1,2.
\end{equation}
We next record a consequence of \eqref{eq:EL_compact} for uniformly bounded states.
\begin{lemma}[Discrete Lipschitz bound for $u^{\hi}_n$]
\label{lem:EL_lipschitz}
Assume \(u=(u^1,u^2) \in \mathcal{K}_N\) solves the discrete Euler--Lagrange equation \eqref{eq:EL_compact}, then \(\|\nabla_{\hi}^+u^{\hi}\|_{\infty, \hi}\) is also uniformly bounded in \(N\) for \({\hi}=1,2\).  Consequently, for every \(n\in I_{\hi}\) and \(m\in \Z \), we have the following discrete Lipschitz bound for \(u^{\hi}\):
\[
|u^{\hi}_{n+m}-u^{\hi}_n|
\leq 3.3 \times 10^4 \pi^3\left(C_V+C_V^2\right) |m| h.
\]
\end{lemma}

We use this discrete Lipschitz bound to compare the full interaction in \eqref{eq:energy_discrete} with a localized form \eqref{eq:energy_localized}. This localized form has an exact vertical-shift symmetry, which yields the following approximate symmetry for the original discrete energy.

\begin{lemma}[Approximate symmetry] \label{lem:approx_symmetry}
Consider the localized energy functional with ${\hi}\in\{1,2\}$:
\begin{equation} \label{eq:energy_localized}
\Eloc^{\hi}[u] := \frac{1}{2} \avsum_{n \in I_{\hi}} (\nabla_{\hi}^{+}u^{\hi}_n)^2 + \avsum_{n \in I_{\hi}}\sum_{m \in \mathbb{Z}}V\left((-1)^{\hj}hn+u^{\hi}_n-u^{\hj}_{n}-2\pi m\right).
\end{equation}
The transformation
\begin{equation} \label{symmetry_localized}
(\tau_{k}u)^1_n = u^1_n+ 2 \pi k_1, \quad (\tau_{k}u)^2_n = u^2_n + 2 \pi k_2, \quad  \forall k_1, k_2 \in \Z,
\end{equation}
satisfies $\Eloc^i[\tau_k u] = \Eloc^i[u]$. For $V \in C^2(\R)$ satisfying \eqref{eq:potential_decay}, $u$ uniformly bounded as in Lemma \ref{lem:uniform_boundedness}, and satisfying \eqref{eq:EL_compact}; define the polynomial for $C_V$ as
\begin{equation}\label{eq:poly_CV}
\mathcal{P}\left(C_V\right):= 9 \times 10^{12} \pi^8  \left(C_V+C_V^5\right) .
\end{equation}
then for $i=1,2$,
$$
\left|E^{\hi}[u] - \Eloc^{\hi}[u]\right| \leq h\mathcal{P}(C_V).
$$
Consequently, for any fixed $k=(k_1,k_2)\in\Z^2$, there exists a constant $C(k,C_V)>0$, independent of $N$, such that
$$
\left|E^i[\tau_k u]-E^i[u]\right|
\le C(k,C_V)h.
$$
\end{lemma}
\begin{remark}
This approximate symmetry is a discrete analogue of the $J[v+2\pi k] = J[v]$ symmetry in the continuum energy. In the continuum limit, this holds exactly, but in the discrete model, for each
fixed $k$, it holds approximately up to an error of order $O(h)$.
\end{remark}
\subsection{Discrete $H^2$ stability}
We now state the discrete stability estimate for the Euler--Lagrange equation
\eqref{eq:EL_compact}. As in the continuum stability result, the estimate
requires a smallness condition on the nonlinear term, formulated below in terms of
the Lipschitz constant of the force map \(F=(F^1,F^2)\).

Let \(v_0\) be the normalized continuum solution of \eqref{eq:EL_v} satisfying $v_0(0)=0$. The corresponding continuum pair is
\[
u_1(x)=\frac{v_0(x)-x}{2}, \qquad
u_2(x)=\frac{x-v_0(x)}{2}.
\]
We denote by \(u_c=(u_c^1,u_c^2)\) the sampling of this continuum pair on the
two grids, namely
\[
(u_c^i)_n := u_i(x_n^i), \qquad n\in I_i, \quad i=1,2.
\]
Here \(x_n^i := n h_i, \, n\in I_i\) denotes the grid point in layer \(i\), with the periodic extension \(x_{n+N_i}^i=x_n^i+2\pi\). Since $v_0$ satisfies $-\pi \leq v_0(x) \leq \pi$ for all $x \in [-\pi,\pi]$, the periodic sampled pair \(u_c\) is uniformly bounded in $\R$: 
\begin{equation} \label{eq:uc_bound}
\left\|u_c^1\right\|_{\infty, 1}+\left\|u_c^2\right\|_{\infty, 2} \leq 2 \pi .
\end{equation}

Let \(u_d=(u_d^1,u_d^2)\) be an exact discrete solution of
\eqref{eq:EL_compact}. In the argument below, \(u_d\) will be the
symmetry-normalized discrete minimizer constructed in Theorem \ref{thm:discrete_minimizer} using Corollary \ref{lem:mean_fixing};
in particular, by Lemma \ref{lem:uniform_boundedness}, \(u_d\)
is also uniformly bounded by:
\begin{equation} \label{eq:ud_bound}
\left\|u_d^1\right\|_{\infty, 1}+\left\|u_d^2\right\|_{\infty, 2} \leq 
42\pi \sqrt{C_V} + 416\pi^2 C_V + \frac{2\pi}{3}.
\end{equation}

The sampled continuum pair \(u_c\) is not an exact solution of the discrete
Euler--Lagrange equation. We define its residual by
\begin{equation} \label{eq:EL_residual}
r^i := \Delta_{\hi} u_c^i - F^i(u_c), \qquad i=1,2.
\end{equation}
The consistency estimate will later give \(\|r\|_{\ell_h^2}=O(h)\).

Define the error
\[
w=(w^1,w^2):=u_d-u_c.
\]
Subtracting \eqref{eq:EL_residual} from \eqref{eq:EL_compact} gives
\begin{equation} \label{eq:error_equation}
\Delta_{\hi} w^i = F^i(u_d)-F^i(u_c)-r^i, \qquad i=1,2.
\end{equation}

\begin{theorem}[Discrete \(H^2\)-stability]
\label{thm:discrete_H2_stability}
Let \(u_c\), \(u_d\), \(r\), and \(w=u_d-u_c\) be defined as above. Let $C_P$ be the Poincaré constant defined in Lemma \ref{lem:two_layer_poincare}, and let \(L_F\) be the Lipschitz constant of the force map \(F=(F^1,F^2)\) defined in Lemma \ref{lem:potential_estimate}. Assume the smallness condition
\begin{equation}
\label{eq:discrete_stability_smallness}
C_P L_F < 1,
\end{equation}
with the gap \(1-C_P L_F\) bounded below independently of \(N\). Then there
exists a constant \(C>0\) depending only on \(C_P\) and \(L_F\), such that the error \(w\) satisfies the discrete \(H^2\)-stability estimate
\begin{equation}
\label{eq:discrete_H2_stability}
\|w\|_{H_h^2}
\leq
\|w\|_{\ell_h^2}
+
\|\nabla_h w\|_{\ell_h^2}
+
\|\Delta_h w\|_{\ell_h^2}
\leq
C\left(\|r\|_{\ell_h^2}+\left(\bigl|(w^1)\bigr|+\bigl|(w^2)\bigr|\right)\right),
\end{equation}
where
$$
C=\frac{1}{1-C_P L_F} \max \left\{1+\sqrt{C_P}+C_P, \sqrt{2 \pi}\left[1+\left(1+\sqrt{C_P}\right) L_F\right]\right\}.
$$
\end{theorem}

The proof follows the error equation \eqref{eq:error_equation}. The constants \(L_F\) and \(C_P\) appearing in the statement are defined in the two lemmas below. The first ingredient is a force estimate \(F(u_d)-F(u_c)\), which uses the uniform boundedness of both states. The second ingredient is a two-layer discrete Poincar\'e inequality, which converts control of \(\Delta_h w\) into control of \(\nabla_h w\) and \(w\), up to the layer means. The smallness condition \(C_P L_F<1\) then allows the \(\ell_h^2\)-term to be absorbed.

To state the force estimate, we connect the sampled continuum state $u_c$ and the exact discrete state $u_d$ by the line segment
\[
z_t := u_c+t(u_d-u_c), \qquad t\in[0,1].
\]
The following lemma is the point where the uniform boundedness results for
\(u_c\) and \(u_d\) enter.

\begin{lemma} \label{lem:potential_estimate}
Assume that \(V''\) satisfies the polynomial decay condition
\eqref{eq:potential_decay}. Let \(u_c\) and \(u_d\) be the normalized states defined above. Then for \(i=1,2\), 
\begin{equation} \label{Assumption}
\max _{i=1,2} \sup _{N \in \mathbb{N}} \sup _{n \in I_i} \sum_{m \in \mathbb{Z}} \sup _{t \in[0,1]}\left|V''\left(A_{n, m}^i\left[z_t\right]\right)\right| \leq 13 C_V\left(1+\frac{8 \pi}{3}+42 \pi C_V^{1 / 2}+416 \pi^2 C_V\right)
\end{equation}
is uniformly bounded, and consequently, the force \(F\) admits a uniform Lipschitz bound with Lipschitz constant $L_F$:
\begin{equation} \label{eq:force_lipschitz}
\left\|F(u_d)-F(u_c)\right\|_{\ell_h^2} \leq L_F \left\|u_d-u_c\right\|_{\ell_h^2},
\end{equation}
where
$$
L_F:= 4\sqrt{2}\max _{i=1,2} \sup _{N \in \mathbb{N}} \sup _{n \in I_i} \sum_{m \in \mathbb{Z}} \sup _{t \in[0,1]}\left|V''\left(A_{n, m}^i\left[z_t\right]\right)\right|.
$$
\end{lemma}
\begin{remark}
The constant $L_F$ is a Lipschitz bound for the nonlinear force $F$ in the discrete Euler--Lagrange equation \eqref{eq:EL_compact}. Like $\Lip(W')=\|W''\|_{L^\infty}$ in the continuum stability estimate of Proposition~\ref{lem:stability}, it controls the force difference $F(u_d)-F(u_c)$ in terms of the displacement error $u_d-u_c$, measured in the weighted $\ell_h^2$ norm.
\end{remark}
The remaining ingredient for the stability estimate is the following two-layer discrete Poincar\'e--Wirtinger inequality.
\begin{lemma}[Two-layer discrete Poincaré--Wirtinger inequality] \label{lem:two_layer_poincare} 
Let \(u=(u^1,u^2)\in \R^N\times \R^{N+1}\), with periodic conventions \(u^{\hi}_{n+N_{\hi}}=u^{\hi}_n\). Then, for each \({\hi}=1,2\), 
\begin{equation} \label{eq:Poincare_layerwise}
\left\|u^{\hi}\right\|^2_{\ell^2_{{\hi}}} \leq C_{P, {\hi}}\left\|\nabla_{\hi}^{+} u^{\hi}\right\|^2_{\ell^2_{{\hi}}}+2 \pi\left|\left(u^{\hi}\right)\right|^2, \quad C_{P, {\hi}} = \frac{h_{\hi}^2}{4 \sin ^2\left(h_{\hi} / 2\right)}
\end{equation}
Consequently, with $C_P:=C_{P,1}$, one obtains \(1 \leq C_P \leq \frac{\pi^2}{4}\) for $N \geq 2$ and 
\[ \|u\|_{\ell_h^2}^2 \le C_P\left( \|\nabla^+_{1}u^1\|_{\ell^2_1}^2 + \|\nabla^+_{2}u^2\|_{\ell^2_2}^2 \right) +2\pi \left(\left(u^1\right)^2+\left(u^2\right)^2\right) . \] 
\end{lemma}


\begin{proof}[Proof of Theorem \ref{thm:discrete_H2_stability}]
Take the product \(\ell_h^2\)-norm of the error equation \eqref{eq:error_equation} and use the force Lipschitz bound \eqref{eq:force_lipschitz}, one obtains
\begin{equation}
\label{eq:laplacian_error_bound}
\|\Delta_h w\|_{\ell_h^2}
\leq
\|F(u_d)-F(u_c)\|_{\ell_h^2}
+
\|r\|_{\ell_h^2}
\leq
L_F\|w\|_{\ell_h^2}
+
\|r\|_{\ell_h^2}.
\end{equation}

Next we apply the Poincaré inequality \eqref{eq:Poincare_layerwise} to \(w^i\) and summing over the two layers gives
\begin{equation}
\label{eq:poincare_for_error}
\|w\|_{\ell_h^2}^2
\leq
C_P\|\nabla_h w\|_{\ell_h^2}^2
+
2\pi \left(\bigl|(w^1)\bigr|^2+\bigl|(w^2)\bigr|^2\right) .
\end{equation}
Since \(\nabla^-_{\hi}w^i\) has zero mean by periodicity, another application
of \eqref{eq:Poincare_layerwise} gives
\[
\|\nabla^-_{\hi}w^i\|_{\ell_i^2}^2
\leq
C_{P,i}
\|\nabla^+_{\hi}\nabla^-_{\hi}w^i\|_{\ell_i^2}^2
=
C_{P,i}
\|\Delta_{\hi}w^i\|_{\ell_i^2}^2 .
\]
we obtain
\begin{equation}
\label{eq:gradient_laplacian_consequence}
\|\nabla_h w\|_{\ell_h^2}^2
\leq
C_P\|\Delta_h w\|_{\ell_h^2}^2 .
\end{equation}
Combining \eqref{eq:poincare_for_error} and
\eqref{eq:gradient_laplacian_consequence}, we get
\[
\|w\|_{\ell_h^2}^2
\leq
C_P^2\|\Delta_h w\|_{\ell_h^2}^2
+
2\pi \left(\bigl|(w^1)\bigr|^2+\bigl|(w^2)\bigr|^2\right),
\]
and hence
\begin{equation}
\label{eq:w_laplacian_mean}
\|w\|_{\ell_h^2}
\leq
C_P\|\Delta_h w\|_{\ell_h^2}
+
\sqrt{2\pi}\left(\bigl|(w^1)\bigr|+\bigl|(w^2)\bigr|\right).
\end{equation}
Substituting \eqref{eq:laplacian_error_bound} into
\eqref{eq:w_laplacian_mean} yields
\[
\|w\|_{\ell_h^2}
\leq
C_P L_F\|w\|_{\ell_h^2}
+
C_P\|r\|_{\ell_h^2}
+
\sqrt{2\pi}\left(\bigl|(w^1)\bigr|+\bigl|(w^2)\bigr|\right).
\]
By the smallness condition \eqref{eq:discrete_stability_smallness}, we can absorb
the first term on the right-hand side and obtain
\begin{equation}
\label{eq:w_stability_bound}
\|w\|_{\ell_h^2} \leq \frac{C_P}{1-C_P L_F}\|r\|_{\ell_h^2}+\frac{\sqrt{2 \pi}}{1-C_P L_F}\left(\left|\left(w^1\right)\right|+\left|\left(w^2\right)\right|\right) .
\end{equation}
From \eqref{eq:laplacian_error_bound} 
\begin{equation}
\label{eq:laplacian_stability_bound}
\left\|\Delta_h w\right\|_{\ell_h^2} \leq \frac{1}{1-C_P L_F}\|r\|_{\ell_h^2}+\frac{\sqrt{2 \pi} L_F}{1-C_P L_F}\left(\left|\left(w^1\right)\right|+\left|\left(w^2\right)\right|\right) .
\end{equation}
From \eqref{eq:gradient_laplacian_consequence}
\begin{equation}
\label{eq:gradient_stability_bound}
\left\|\nabla_h w\right\|_{\ell_h^2} \leq \frac{\sqrt{C_P}}{1-C_P L_F}\|r\|_{\ell_h^2}+\frac{\sqrt{2 \pi C_P} L_F}{1-C_P L_F}\left(\left|\left(w^1\right)\right|+\left|\left(w^2\right)\right|\right) .
\end{equation}
Adding the three estimates gives
$$
\begin{aligned}
\|w\|_{\ell_h^2}+\left\|\nabla_h w\right\|_{\ell_h^2}+\left\|\Delta_h w\right\|_{\ell_h^2} & \leq \frac{1+\sqrt{C_P}+C_P}{1-C_P L_F}\|r\|_{\ell_h^2} \\
&+\frac{\sqrt{2 \pi}\left[1+\left(1+\sqrt{C_P}\right) L_F\right]}{1-C_P L_F}\left(\left|\left(w^1\right)\right|+\left|\left(w^2\right)\right|\right).
\end{aligned}
$$
By taking $C$ as
$$
C=\frac{1}{1-C_P L_F} \max \left\{1+\sqrt{C_P}+C_P, \sqrt{2 \pi}\left[1+\left(1+\sqrt{C_P}\right) L_F\right]\right\},
$$
we obtain the desired $H^2$-stability estimate \eqref{eq:discrete_H2_stability}.
\end{proof}

\section{Consistency of the Euler--Lagrange equation and convergence} \label{sec:consistency_convergence}
We now show that the sampled continuum solution \(u_c\) is an approximate
solution of the discrete Euler--Lagrange equation. Throughout this section we
assume that the continuum potential \(\Phi\) is generated by the atomistic potential
\(V\) through the periodization
\begin{equation}
\label{eq:Phi_periodization}
\Phi(s) := 2\sum_{m\in\Z} V(s-2\pi m).
\end{equation}
The factor \(2\) comes from the two layer contributions in the discrete energy. Since $V\in C^2(\R)$ is even and satisfies the decay condition \eqref{eq:potential_decay}, the series for $\Phi$ and its first two derivatives converge uniformly on $[-\pi,\pi]$. Thus $\Phi$ is even, $2\pi$--periodic and $C^2$, and hence satisfies the assumptions for Theorem \ref{thm:v}.
\begin{theorem}[Consistency of the discrete Euler--Lagrange equation]
\label{thm:EL_consistency}
Let \(v_0\) be the normalized continuum solution of \eqref{eq:EL_v} with $v_0(0)=0$, and let
\(u_c=(u_c^1,u_c^2)\) be its sampling on the two grids. Let \(r\) be
the residual defined by \eqref{eq:EL_residual}. Then
\begin{equation}
\label{eq:EL_consistency_pointwise}
\left\|r\right\|_{\infty}\leq 1300(C_V^{3/2}+ C_V) h,
\end{equation}
and
\begin{equation}
\label{eq:EL_consistency_l2}
\|r\|_{\ell_h^2}\leq 5000(C_V^{3/2}+ C_V) h.
\end{equation}
\end{theorem}

\begin{proof}[Proof of Theorem \ref{thm:EL_consistency}]
First by the decay assumption \eqref{eq:potential_decay} and the periodization \eqref{eq:Phi_periodization}, the series for \(\Phi,\Phi',\Phi''\) converge uniformly and one obtains a convenient relaxed bound as:
\begin{equation}
\label{eq:Phi_bound}
\|\Phi\|_{L^{\infty}},\left\|\Phi'\right\|_{L^{\infty}},\left\|\Phi''\right\|_{L^{\infty}} \leq 4 C_V .
\end{equation}
Choose $x_*\in[-\pi,\pi]$ at which $v_0'$ attains its minimum. Since $v_0'>0$ and
$$
\frac{1}{2\pi}\int_{-\pi}^{\pi}v_0'(x)\,\text{d}x=1,
$$
we have $0<v_0'(x_*)\leq1$. With $x_*$ fixed, the first integral for \eqref{eq:EL_v} gives, for every $x\in[-\pi,\pi]$,
$$
\left(v_0'(x)\right)^2
=\left(v_0'(x_*)\right)^2+4\left(\Phi\left(v_0(x)\right)-\Phi\left(v_0(x_*)\right)\right)
\leq1+8\|\Phi\|_{L^{\infty}}.
$$
Consequently,
\begin{equation}
\label{eq:v0_derivative_bound}
\left\|v_0'\right\|_{L^{\infty}}^2 \leq 1+8 \|\Phi\|_{L^{\infty}} \leq 1+32 C_V .
\end{equation}
Recall that the continuum pair is defined by
\[
u_1(x)=\frac{v_0(x)-x}{2},
\qquad
u_2(x)=\frac{x-v_0(x)}{2},
\]
so that
\[
v_0(x)=x+u_1(x)-u_2(x).
\]
Apply the above estimate \eqref{eq:v0_derivative_bound} one obtains
\begin{equation}
\label{eq:continuum_pair_derivative_bound}
\left\|\left(u_i\right)'\right\|_{L^{\infty}} \leq 3\left(\sqrt{C_V}+1\right).
\end{equation}
Since \(v_0\) solves \eqref{eq:EL_v}, the continuum pair satisfies
\begin{equation}
\label{eq:continuum_EL_pair}
(u_i)''(x)=(-1)^{3-i}\Phi'(v_0(x)), \qquad i=1,2.
\end{equation}
Combine \eqref{eq:continuum_EL_pair} with \eqref{eq:v0_derivative_bound} and \eqref{eq:Phi_bound} one obtains
\begin{equation}
\label{eq:continuum_pair_lipschitz}
\Lip\left(\left(u_i\right)''\right) \leq\left\|\Phi''\right\|_{L^{\infty}}\left\|v_0'\right\|_{L^{\infty}} \leq 23 (C_V^{3/2}+ C_V) .
\end{equation}

We now compare the full discrete interaction argument \(\Anm[u_c]\) with its
localized version. Define
\begin{equation}
\label{eq:Bnm_consistency}
\Bnm[u_c] := (-1)^{3-i}hn+(u_c^i)_n-(u_c^{3-i})_n-2\pi m, \quad n\in I_i,\ m\in\Z,
\end{equation}
where all indices are understood periodically, as before. Let
\[
\rho_{n,m}^i:=\Anm[u_c]-\Bnm[u_c].
\]
Using \eqref{eq:argument} and \eqref{eq:Bnm_consistency}, we obtain
\[
\begin{aligned}
\rho_{n,m}^i &=(-1)^{3-i}\frac{h}{4\pi}\left((u_c^i)_n+(u_c^{3-i})_n\right)\\
&-\left(1+(-1)^{i}\frac{h}{4\pi}\right)\left((u_c^{3-i})_{n+m}-(u_c^{3-i})_n\right)
-
2\pi(-1)^{i}\frac{h}{4\pi} m.
\end{aligned}
\]
Apply \eqref{eq:continuum_pair_derivative_bound} to the sampled continuum pair $u_c$, we obtain
\begin{equation} \label{eq:continuum_pair_difference}
|(u_c^{\hi})_{n+m}-(u_c^{\hi})_n|\leq \|(u_{\hi})'\|_{L^{\infty}}|x_{n+m}^{\hi}-x_n^{\hi}|
\leq 3\left(\sqrt{C_V}+1\right)h_{\hi}|m|,
\quad i=1,2.
\end{equation}
Therefore, combine \eqref{eq:continuum_pair_difference} and \eqref{eq:uc_bound}, one obtains
\begin{equation}
\label{eq:rho_consistency_bound}
\begin{aligned}
\left|\rho_{n, m}^i\right| & \leq \frac{h}{2}+(1+\frac{h}{4\pi}) 3\left(\sqrt{C_V}+1\right)h_{\hj}|m|+\frac{h}{2}|m| \\
& \leq 7\left(\sqrt{C_V}+1\right)(1+|m|) h .
\end{aligned}
\end{equation}

We claim that
\begin{equation}
\label{eq:force_localization_consistency}
\left|
2\sum_{m\in\Z}
\left[
V'(\Anm[u_c])-V'(\Bnm[u_c])
\right]
\right|
\leq 1260\left(C_V+C_V^{3 / 2}\right) h.
\quad i=1,2.
\end{equation}
Set
$$
C_{n, m, \theta}^i:=B_{n, m}^i\left[u_c\right]+\theta \rho_{n, m}^i=(1-\theta) B_{n, m}^i\left[u_c\right]+\theta A_{n, m}^i\left[u_c\right], 
$$
then
$$
\left|C_{n, m, \theta}^i\right| \geq \frac{4 \pi}{3}|m|-5 \pi .
$$
Note that here the bound we used is different from the one we used in the proof of Lemma \ref{lem:approx_symmetry}. This is due to the fact that we are now considering the sampled continuum pair \(u_c\), which has a cleaner uniform bound \eqref{eq:uc_bound}. For $|m| \geq 9$,
$$
\left|C_{n, m, \theta}^i\right| \geq \frac{2 \pi}{3}|m| .
$$
The near-field part $|m| \leq 8$ and the far-field decay give the weighted estimate
$$
\sup _{i, n} \sum_{m \in \mathbb{Z}}(1+|m|) \sup _{\theta \in[0,1]}\left|V''\left(C_{n, m, \theta}^i\right)\right| \leq 90 C_V .
$$
It follows that, by the mean value theorem,
$$
\begin{aligned}
\left|2 \sum_m\left[V'\left(A_{n, m}^i\left[u_c\right]\right)-V'\left(B_{n, m}^i\left[u_c\right]\right)\right]\right| &\leq 2 \sum_m \sup _\theta\left|V''\left(C_{n, m, \theta}^i\right)\right|\left|\rho_{n, m}^i\right|  \\
&\leq 1260\left(C_V+C_V^{3 / 2}\right) h.
\end{aligned}
$$
For \(n\in I_i\), set
\[
s_n := hn+(u_c^1)_n-(u_c^2)_n.
\]
Then
\[
\Bnm[u_c]=(-1)^{3-i}s_n-2\pi m.
\]
Using the evenness of \(V\), so that \(V'\) is odd, and applying the periodization \eqref{eq:Phi_periodization}, we obtain
\[
2\sum_{m\in\Z}V'(\Bnm[u_c]) = (-1)^{3-i}\Phi'(s_n).
\]
Thus, for \(i=1,2\), we have
\begin{equation}
\label{eq:force_to_phi_consistency}
\left|F_n^i\left(u_c\right)-(-1)^{3-i} \Phi'\left(s_n\right)\right| \leq 1260\left(C_V+C_V^{3 / 2}\right) h.
\end{equation}

It remains to estimate the phase error. Compare \(s_n\) with the continuum phase at the grid point \(x_n^i\). Apply \eqref{eq:continuum_pair_derivative_bound}, we obtain
\begin{equation}
\label{eq:phase_consistency}
\begin{aligned}
|s_n-v_0(x_n^i)| & \leq |h-h_i|n + |u_1(x_n^1)-u_1(x_n^i)| + |u_2(x_n^2)-u_2(x_n^i)| \\
& \leq \frac{h}{2} + \left\|\left(u_{\hj}\right)'\right\|_{L^{\infty}}|x_n^{\hj}-x_n^i|  
\leq 4\left(\sqrt{C_V}+1\right) h.
\end{aligned}
\end{equation}
By \eqref{eq:Phi_bound}, \(\Lip{(\Phi')}\leq 4 C_V\), and hence
\begin{equation}
\label{eq:phi_prime_consistency}
\left|\Phi'\left(s_n\right)-\Phi'\left(v_0\left(x_n^i\right)\right)\right| \leq 16(C_V^{3/2}+ C_V) h.
\end{equation}

On the other hand, apply the central-difference expansion with integral remainder and combine with \eqref{eq:continuum_pair_lipschitz}, we obtain
\begin{equation}
\label{eq:continuum_laplacian_consistency}
\left|\Delta_{\hi}u_c^i - (u_i)''(x_n^i)\right| \leq \frac{h_i}{3} \Lip\left(\left(u_i\right)''\right) \leq 12(C_V^{3/2}+ C_V) h.
\end{equation}
Combine \eqref{eq:continuum_EL_pair}, \eqref{eq:force_to_phi_consistency}, \eqref{eq:phi_prime_consistency}, and \eqref{eq:continuum_laplacian_consistency}, we obtain
$$
\sup_{n,i} \left|r^i_n\right| \leq 1300(C_V^{3/2}+ C_V) h.
$$
This proves \eqref{eq:EL_consistency_pointwise}. Finally,
\[
\|r\|_{\ell_h^2} = \sqrt{\sum_{i=1}^2 h_i\sum_{n\in I_i}|r_n^i|^2} \leq 5000(C_V^{3/2}+ C_V) h.
\]
\end{proof}

We can now combine the consistency estimate with the stability estimate to obtain the main convergence result.


\begin{theorem}[Discrete-to-continuum convergence]
\label{thm:discrete_to_continuum_convergence}
For $N \geq 2$, let $V \in C^2(\R)$ be an even potential satisfying the decay condition \eqref{eq:potential_decay}, and let the continuum misfit potential $\Phi$ be defined by \eqref{eq:Phi_periodization}.Assume that
\[
1-C_PL_F>0,
\]
with the gap bounded below independently of $N$.
Then for \(u_d\) the symmetry-normalized discrete minimizer and \(u_c\) the sampled normalized continuum minimizer, we have the following convergence result in the discrete \(H^2\) norm:
\begin{equation}
\label{eq:convergence_result}
\|u_d-u_c\|_{H_h^2} \leq C_{\mathrm{conv}}\left(C_V, C_P, L_F\right) h.
\end{equation}
where
\begin{equation}
\label{eq:convergence_constant}
\begin{aligned}
C_{\mathrm{conv}}\left(C_V, C_P, L_F\right):=\frac{5000}{1-C_P L_F} & {\left[\left(1+\sqrt{C_P}+C_P\right)\left(C_V+C_V^{3 / 2}\right)\right.} \\
& \left.+\left(1+\left(1+\sqrt{C_P}\right) L_F\right)\right] .
\end{aligned}
\end{equation}
\end{theorem}
\begin{remark}
The smallness condition $C_PL_F<1$ is analogous to $\Lip(W')<1$ in Proposition~\ref{lem:stability}. In both proofs, the Poincar\'e--Wirtinger inequality controls the displacement error by its Laplacian, up to the mean values. The Euler--Lagrange error equation and the force Lipschitz bound then yield $H^2$-stability by absorbing the nonlinear term under the smallness condition. Moreover, $C_P\to1$ as $N\to\infty$, agreeing with the continuum Poincar\'e constant for $2\pi$-periodic functions. Thus both smallness conditions ensure stability through the same absorption argument, with different force Lipschitz bounds.
\end{remark}
\begin{proof}[Proof of Theorem \ref{thm:discrete_to_continuum_convergence}]
Let $w:=u_d-u_c$. By Theorem \ref{thm:EL_consistency},
\[
\|r\|_{\ell_h^2}\leq 5000(C_V^{3/2}+ C_V) h.
\]
It remains to estimate the mean term in the stability bound.

Since the normalized continuum minimizer \(v_0\) is odd and satisfies $v_0(x+2\pi)=v_0(x) + 2\pi$, the associated pair
\[
u_1(x)=\frac{v_0(x)-x}{2},
\qquad
u_2(x)=\frac{x-v_0(x)}{2}
\]
are both odd and 2$\pi$-periodic, satisfying $u_i(2 \pi-x)=u_i(-x)=-u_i(x)$. Hence, for $1 \leq n \leq N_i-1$,
$$
u_i\left(x_{N_i-n}^i\right)=u_i\left(2 \pi-x_n^i\right)=-u_i\left(x_n^i\right) .
$$
The index $n=0$ contributes $u_i(0)=0$, and, when $N_i$ is even, the additional fixed index $n=N_i / 2$ contributes $u_i(\pi)=0$. Therefore $\left(u_c^i\right)=0$ for $i=1,2$.

For the symmetry-normalized discrete minimizer, by Corollary \ref{lem:mean_fixing}, we have
\[
|(u_d^i)|\leq \frac{h}{4},
\qquad i=1,2.
\]
So for $w=u_d-u_c$,
$$
\left|\left(w^1\right)\right|+\left|\left(w^2\right)\right|=\left|\left(u_d^1\right)\right|+\left|\left(u_d^2\right)\right| \leq \frac{h}{2} .
$$
Adding the estimates \eqref{eq:w_stability_bound}, \eqref{eq:laplacian_stability_bound}, and \eqref{eq:gradient_stability_bound}, and using the preceding residual and mean bounds, we obtain 
\[
\|w\|_{\ell_h^2} + \|\nabla_h w\|_{\ell_h^2} + \|\Delta_h w\|_{\ell_h^2}
\leq 
C_{\mathrm{conv}}\left(C_V, C_P, L_F\right) h,
\]
where $C_{\mathrm{conv}}$ is defined as in \eqref{eq:convergence_constant}. The quantity on the left is an equivalent discrete $H_h^2$ norm and bounds $\|w\|_{H_h^2}$ from above, as noted in Section \ref{sec:discrete_min}. This gives \eqref{eq:convergence_result} with the same constant. 
\end{proof}
Using the discrete Sobolev embedding as in Appendix \ref{appendix:lem3_proof}, we can immediately obtain the following convergence result in the $\ell^\infty$ norm.
\begin{corollary}[Convergence in the discrete maximum norm]
\label{cor:convergence_inf_norm}
Under the assumptions of Theorem \ref{thm:discrete_to_continuum_convergence}, we have
$$
\left\|u_d-u_c\right\|_{\infty} \leq C_{\infty}\left(C_V, C_P, L_F\right) h,
$$
where
$$
C_{\infty}\left(C_V, C_P, L_F\right):=\frac{1}{4}+\sqrt{2 \pi} C_{\mathrm{conv}}\left(C_V, C_P, L_F\right)
$$
\end{corollary}
\begin{proof}
Set $w:=u_d-u_c$. For each layer $i$, following the proof as in Appendix \ref{appendix:lem3_proof}, the discrete 1D-Sobolev embedding gives
$$
\left\|w^i\right\|_{\infty, i} \leq\left|\left(w^i\right)\right|+\sqrt{2 \pi}\left\|\nabla_{h_i} w^i\right\|_{\ell_i^2} .
$$
The sampled continuum means vanish exactly, while the normalized discrete minimizer satisfies $\left|\left(u_d^i\right)\right| \leq \frac{h}{4}$. Thus
$$
\left|\left(w^i\right)\right|=\left|\left(u_d^i\right)-\left(u_c^i\right)\right| \leq \frac{h}{4} .
$$
Taking the maximum over $i=1,2$,
$$
\|w\|_{\infty} \leq \frac{h}{4}+\sqrt{2 \pi}\left\|\nabla_h w\right\|_{\ell_h^2} \leq \left(\frac{1}{4}+\sqrt{2 \pi} C_{\mathrm{conv}}\right) h,
$$
where the last inequality follows from the convergence theorem \ref{thm:discrete_to_continuum_convergence}. This proves the result.
\end{proof}


\subsection{\texorpdfstring{$\Gamma$}{Gamma}-convergence without the smallness condition}\label{sec:gamma}
The rate in Theorem \ref{thm:discrete_to_continuum_convergence} relies on the smallness condition \eqref{eq:discrete_stability_smallness} through the stability estimate of Theorem \ref{thm:discrete_H2_stability}. We now show that, without \eqref{eq:discrete_stability_smallness}, the discrete energies still $\Gamma$-converge to the continuum energy (see Appendix \ref{appendix:gamma} for the relevant definitions). As a consequence, minimizers still converge, but without a rate. Throughout this subsection, $V\in C^2(\R)$ is even and satisfies \eqref{eq:potential_decay}, and $\Phi$ is given by \eqref{eq:Phi_periodization}.

To compare discrete and continuum configurations, we use piecewise affine interpolation on the grid points $x_n^i=nh_i$ (with $x_{n+N_i}^i=x_n^i+2\pi$). For a discrete pair $u=(u^1,u^2)$, let $\hat u=(\hat u_1,\hat u_2)\in\Hp\times\Hp$ be the pair such that $\hat u_i$ is affine on each interval $[x_n^i,x_{n+1}^i]$ and $\hat u_i(x_n^i)=u_n^i$. Since $\hat u_i'=\fdi u^i_n$ on $(x_n^i,x_{n+1}^i)$, we have
\begin{equation}\label{eq:interpolant}
\|\hat u_i'\|_{L^2(-\pi,\pi)}=\|\fdi u^i\|_{\ell^2_i},
\qquad
\|\hat u_i\|_{L^\infty}=\|u^i\|_{\infty,i},
\qquad
(\hat u_i)=(u^i),
\end{equation}
where the last identity is the periodic trapezoidal rule. In particular, the elastic parts of $E_N[u]$ and $E[\hat u]$ coincide. Moreover, the Cauchy--Schwarz inequality and periodicity show that every $g\in\Hp$ satisfies
\begin{equation}\label{eq:holder}
|g(x)-g(y)|\le|x-y|^{1/2}\,\|g'\|_{L^2(-\pi,\pi)},\qquad x,y\in\R.
\end{equation}
Conversely, $u$ is recovered from $\hat u$ by evaluation at the grid points,
$u_n^i=\hat u_i(x_n^i)$. To compare $E_N$ and $E$ on a common space, we regard
$E_N$ as a functional on $L^2(-\pi,\pi)^2$. For $f\in L^2(-\pi,\pi)^2$ we set
$E_N[f]:=E_N[u]$ if $f=\hat u$ for some discrete pair $u$, that is, if each
component $f_i$ is continuous, $2\pi$-periodic and affine on every interval
$[x_n^i,x_{n+1}^i]$ of its own grid. Otherwise we set $E_N[f]:=+\infty$; this
includes functions that are piecewise affine with respect to other grids.
Likewise, we set $E:=+\infty$ on $L^2(-\pi,\pi)^2\setminus(\Hp\times\Hp)$.

The only new estimate needed is the following variant of Lemma \ref{lem:approx_symmetry}, which compares the discrete and continuum energies on bounded sets.
\begin{lemma}[Energy consistency]\label{lem:energy_consistency}
For every $R>0$ there exists a constant $C_R>0$, depending only on $R$ and $C_V$, such that every discrete pair $u$ with $\|u\|_\infty+\|\nabla_h u\|_{\ell_h^2}\le R$ satisfies
\[
\bigl|E_N[u]-E[\hat u]\bigr|\le C_R\sqrt h .
\]
In contrast to Lemma \ref{lem:approx_symmetry}, $u$ is not required to satisfy \eqref{eq:EL_compact}.
\end{lemma}

\begin{proof}
Throughout the proof, $C_R$ denotes a constant depending only on $R$ and $C_V$, whose value may change from line to line. As in the proof of Lemma \ref{lem:approx_symmetry}, let $U:=\|u^1\|_{\infty,1}+\|u^2\|_{\infty,2}$, so that $U\le2R$. We use repeatedly that $h_1,h_2\le2h$ and $h\le4\pi/5$, so that $h\le2\sqrt h$.

\emph{Step 1: localization.} Let $\Eloc^i$ be the localized energy \eqref{eq:energy_localized} and let $\Bnm[u]$ be as in the proof of Lemma \ref{lem:approx_symmetry}. Before Lemma \ref{lem:EL_lipschitz} is used, that proof shows that
\[
\bigl|\Anm[u]-\Bnm[u]\bigr|\le\frac{h}{4\pi}U+\Bigl(1+\frac{h}{4\pi}\Bigr)\bigl|u^{\hj}_{n+m}-u^{\hj}_n\bigr|+\frac h2|m| .
\]
The equation \eqref{eq:EL_compact} enters that proof only through the Lipschitz bound of Lemma \ref{lem:EL_lipschitz}. Here we replace that bound with \eqref{eq:holder}, applied to $\hat u_{\hj}$ at the grid points $x^{\hj}_{n+m}$ and $x^{\hj}_n$:
\[
\bigl|u^{\hj}_{n+m}-u^{\hj}_n\bigr|\le(|m|h_{\hj})^{1/2}\,\|\nabla_h u\|_{\ell_h^2}\le2R\sqrt h\,|m| .
\]
Hence $|\Anm[u]-\Bnm[u]|\le C_R\sqrt h\,(1+|m|)$. The bounds \eqref{eq:Vprime_sum} and \eqref{eq:Vprime_msum} use only the lower bound $|C^i_{n,m,\theta}|\ge\pi|m|-2\pi(1+U)$, which holds for every discrete pair, so they remain valid. By the mean value theorem, as in \eqref{eq:V_est}, and since $h_iN_i=2\pi$, we obtain for $i=1,2$
\begin{equation}\label{eq:loc_sqrt_h}
\bigl|E^i[u]-\Eloc^i[u]\bigr|
\le C_R\sqrt h\;h_i\sum_{n\in I_i}\sum_{m\in\Z}(1+|m|)\sup_{0\le\theta\le1}\bigl|V'(C^i_{n,m,\theta})\bigr|
\le C_R\sqrt h .
\end{equation}

\emph{Step 2: comparison with the continuum interaction.} For $n\in I_i$, set $s_n:=hn+u^1_n-u^2_n$, as in the proof of Theorem \ref{thm:EL_consistency}, so that $\Bnm[u]=(-1)^{\hj}s_n-2\pi m$. Since $V$ is even, \eqref{eq:Phi_periodization} gives $\sum_{m\in\Z}V(\Bnm[u])=\frac12\Phi(s_n)$ for $i=1,2$, and therefore
\[
\Eloc^1[u]+\Eloc^2[u]=\frac12\|\nabla_h u\|_{\ell_h^2}^2+\frac12\sum_{i=1}^2h_i\sum_{n\in I_i}\Phi(s_n).
\]
Let $v(x):=x+\hat u_1(x)-\hat u_2(x)$, as in \eqref{eq:uv_change}. Since $u^i_n=\hat u_i(x_n^i)$, the argument leading to \eqref{eq:phase_consistency} applies with \eqref{eq:holder} in place of \eqref{eq:continuum_pair_derivative_bound}. Using $|x_n^{\hj}-x_n^i|\le h_2\le h$, it gives
\[
|s_n-v(x_n^i)|\le\frac h2+\bigl|\hat u_{\hj}(x_n^{\hj})-\hat u_{\hj}(x_n^i)\bigr|\le\frac h2+R\,|x_n^{\hj}-x_n^i|^{1/2}\le C_R\sqrt h .
\]
Moreover, $x\mapsto\Phi(v(x))$ is $2\pi$-periodic. By \eqref{eq:Phi_bound} and \eqref{eq:holder}, it satisfies $|\Phi(v(x))-\Phi(v(y))|\le4C_V\bigl(|x-y|+2R|x-y|^{1/2}\bigr)$. We compare the Riemann sums with the integral cell by cell, using $\|\Phi'\|_{L^\infty}\le4C_V$ once more to replace $s_n$ by $v(x_n^i)$. This gives
\[
\Bigl|\,h_i\sum_{n\in I_i}\Phi(s_n)-\int_{-\pi}^{\pi}\Phi\bigl(v(x)\bigr)\,dx\Bigr|\le C_R\sqrt h,\qquad i=1,2.
\]

\emph{Conclusion.} The interaction part of $E[\hat u]$ in \eqref{eq:E12} is $\int_{-\pi}^{\pi}\Phi(v(x))\,dx$, and by \eqref{eq:interpolant} the elastic parts of $E_N[u]$ and $E[\hat u]$ coincide. Combining the last two displays with \eqref{eq:loc_sqrt_h}, we obtain $|E_N[u]-E[\hat u]|\le C_R\sqrt h$.
\end{proof}

We can now state our main Gamma convergence result. It is important to point out that the smallness condition \eqref{eq:discrete_stability_smallness} is not assumed.
\begin{theorem}[$\Gamma$-convergence and convergence of minimizers]\label{thm:gamma}
Let $V\in C^2(\R)$ be even and satisfy \eqref{eq:potential_decay}, and let $E$ be the continuum energy \eqref{eq:E12} with $\Phi$ given by \eqref{eq:Phi_periodization}. Then the following hold as $N\to\infty$.
\begin{enumerate}
\item[(i)] $E_N$ $\Gamma$-converges to $E$ in $L^2(-\pi,\pi)^2$. Moreover, the restrictions of $E_N$ to $\AN$ are equicoercive, and they $\Gamma$-converge in $L^2(-\pi,\pi)^2$ to the restriction of $E$ to the pairs $u$ with $(u_1)=(u_2)=0$.
\item[(ii)] $\min E_N\to\min E$. Moreover, let $u_d\in\AN$ be a symmetry-normalized discrete minimizer and $u_c$ the sampled normalized continuum minimizer, as in Theorem \ref{thm:discrete_to_continuum_convergence}. Then $\hat u_d$ converges strongly in $H^1(-\pi,\pi)^2$ to the continuum pair sampled by $u_c$, and
\[
\|u_d-u_c\|_{H_h^2}\to0
\qquad\text{and}\qquad
\|u_d-u_c\|_\infty\to0 .
\]
\end{enumerate}
\end{theorem}

\begin{proof}
\emph{Coercivity.} Consider the steps of Appendix \ref{appendix:lem3_proof} that come before the normalization and the bound $E_N[u]\le C_0$ are used. They show that every discrete pair $u$ satisfies
\[
\|u^1\|_{\infty,1}+\|u^2\|_{\infty,2}\le|(u^1)|+|(u^2)|+2\sqrt\pi\,\|\nabla_h u\|_{\ell_h^2}
\]
and
\[
E_N[u]\ge\frac12\|\nabla_h u\|_{\ell_h^2}^2-52\pi C_V\bigl(1+\|u^1\|_{\infty,1}+\|u^2\|_{\infty,2}\bigr).
\]
Combining the two and using Young's inequality, we obtain, for every discrete pair $u$,
\begin{equation}\label{eq:gamma_coercive}
E_N[u]\ge\frac14\|\nabla_h u\|_{\ell_h^2}^2-C\bigl(1+|(u^1)|+|(u^2)|\bigr),
\end{equation}
where $C$ depends only on $C_V$. Hence, if a sequence of discrete pairs $u_N$ has bounded energies and bounded layer means, then $\|u_N\|_\infty+\|\nabla_h u_N\|_{\ell_h^2}$ is bounded. By \eqref{eq:interpolant}, $(\hat u_N)$ is then bounded in $H^1(-\pi,\pi)^2$, and Lemma \ref{lem:energy_consistency} applies with a fixed $R$. No smallness condition is used in this step.

\emph{Proof of (i): liminf inequality.} Consider a sequence that converges to some $u$ in $L^2(-\pi,\pi)^2$. The means of the components of its elements converge, so they are bounded. Elements of finite energy are interpolants $\hat u_N$ of discrete pairs $u_N$, and by \eqref{eq:interpolant} the layer means of $u_N$ are exactly these means. Hence \eqref{eq:gamma_coercive} bounds the energies from below, so their $\liminf$ is not $-\infty$. If it is $+\infty$, there is nothing to prove. Otherwise, after passing to a subsequence along which the $\liminf$ is attained, we may assume that the energies are bounded. Then every element of the sequence is an interpolant $\hat u_N$, and by the coercivity step, $(\hat u_N)$ is bounded in $H^1(-\pi,\pi)^2$. After passing to a further subsequence, $\hat u_N\rightharpoonup u$ weakly in $H^1(-\pi,\pi)^2$; in particular, $u\in\Hp\times\Hp$. Lemma \ref{lem:energy_consistency} gives $E_N[u_N]-E[\hat u_N]\to0$. The elastic part of $E$ is weakly lower semicontinuous in $H^1$. By \eqref{eq:Phi_bound}, the interaction part satisfies
\begin{equation}\label{eq:interaction_lipschitz}
\Bigl|\int_{-\pi}^{\pi}\Phi(x+f_1-f_2)\,dx-\int_{-\pi}^{\pi}\Phi(x+g_1-g_2)\,dx\Bigr|
\le4\sqrt{2\pi}\,C_V\sum_{i=1}^2\|f_i-g_i\|_{L^2(-\pi,\pi)}
\end{equation}
for all $f,g\in L^2(-\pi,\pi)^2$, so it is continuous in $L^2$. Hence
\[
E[u]\le\liminf_{N\to\infty}E[\hat u_N]=\liminf_{N\to\infty}E_N[u_N].
\]

\emph{Proof of (i): recovery sequence.} If $u\notin\Hp\times\Hp$, then $E[u]=+\infty$ and there is nothing to prove, so let $u\in\Hp\times\Hp$. Let $u_N$ be the sampling of $u$ on the two grids, defined in the same way as $u_c$, that is, $u_N$ has entries $u_i(x_n^i)$. On each interval $[x_n^i,x_{n+1}^i]$, the derivative of the $i$-th component of $\hat u_N$ equals the average of $u_i'$ over that interval. Such averages converge to $u_i'$ in $L^2$, and $\hat u_N\to u$ uniformly by \eqref{eq:holder}. Hence $\hat u_N\to u$ in $H^1(-\pi,\pi)^2$, and $\|u_N\|_\infty+\|\nabla_h u_N\|_{\ell_h^2}\le\max_i\|u_i\|_{L^\infty}+\|u'\|_{L^2}$. Lemma \ref{lem:energy_consistency} gives $E_N[u_N]-E[\hat u_N]\to0$. The elastic part of $E$ is continuous in $H^1$, and the interaction part is continuous in $L^2$ by \eqref{eq:interaction_lipschitz}, so $E[\hat u_N]\to E[u]$. Hence $E_N[u_N]\to E[u]$.

\emph{Proof of (i): restricted functionals.} Every $u\in\AN$ has layer means bounded by $h/4$. Hence, by the coercivity step, for each $c\in\R$ the interpolants of all $u\in\AN$ with $E_N[u]\le c$ lie in a closed ball of $H^1(-\pi,\pi)^2$ whose radius depends on $c$ but not on $N$. By the compact embedding $H^1(-\pi,\pi)\hookrightarrow L^2(-\pi,\pi)$, such a ball is compact in $L^2(-\pi,\pi)^2$, so the restricted functionals are equicoercive. If $u_N\in\AN$ and $\hat u_N\to u$ in $L^2(-\pi,\pi)^2$, then \eqref{eq:interpolant} gives $(u_1)=(u_2)=0$, and the liminf inequality follows from the one proved above. For the recovery sequence, let $(u_1)=(u_2)=0$ and let $u_N$ be the sampling of $u$ as above. By \eqref{eq:interpolant}, the layer means of $u_N$ converge to $(u_1)=(u_2)=0$. Subtracting from each layer of $u_N$ its mean therefore produces pairs in $\AN$ whose interpolants still converge to $u$ in $H^1(-\pi,\pi)^2$, and the argument of the recovery step applies.

\emph{Proof of (ii): identification of the limit.} By the reduction of $E$ to $J$ in Section~\ref{sec:cont_min} and Theorem~\ref{thm:v}, a pair $u$ minimizes $E$ if and only if $u_+$ is constant and $v=\tau_dv_0+2\pi k$ for some $d\in\R$, $k\in\Z$. As noted in the proof of Theorem~\ref{thm:discrete_to_continuum_convergence}, $v_0(x)-x$ is odd and $2\pi$-periodic, so $(v-x)=2\pi k-d$. Since $u_1=\frac12(u_++v-x)$ and $u_2=\frac12(u_+-v+x)$, the condition $(u_1)=(u_2)=0$ forces $u_+=0$ and $d=2\pi k$, and then $v=v_0$ because $v_0(x-2\pi k)=v_0(x)-2\pi k$. Hence $E$ has exactly one minimizer $u$ with $(u_1)=(u_2)=0$, namely the continuum pair sampled by $u_c$, and $\min E$ is also the minimum of $E$ over pairs with vanishing layer means. By Theorem \ref{thm:discrete_minimizer}, $\min E_N=\min_{\AN}E_N$. The fundamental theorem of $\Gamma$-convergence (Appendix \ref{appendix:gamma}), applied to the restricted functionals in (i), therefore gives $\min E_N\to\min E$. It also shows that $(\hat u_d)$ is precompact in $L^2(-\pi,\pi)^2$ and that each of its limit points is a minimizer of $E$ with vanishing layer means, that is, equal to $u$. Hence the whole sequence $\hat u_d$ converges to $u$ in $L^2(-\pi,\pi)^2$.

\emph{Proof of (ii): strong $H^1$ convergence.} By the coercivity step, $(\hat u_d)$ is bounded in $H^1(-\pi,\pi)^2$, so $\hat u_d\rightharpoonup u$ weakly in $H^1(-\pi,\pi)^2$. Lemma \ref{lem:energy_consistency} gives $E[\hat u_d]-E_N[u_d]\to0$, so $E[\hat u_d]\to\min E=E[u]$. By \eqref{eq:interaction_lipschitz}, the interaction parts also converge, and hence $\|\hat u_d'\|_{L^2}\to\|u'\|_{L^2}$. Together with the weak convergence, this gives $\hat u_d\to u$ strongly in $H^1(-\pi,\pi)^2$.

\emph{Proof of (ii): discrete $H^2$ convergence.} Since $H^1(-\pi,\pi)\hookrightarrow C([-\pi,\pi])$, the convergence $\hat u_d\to u$ is also uniform. The entries of $u_d$ and $u_c$ are the nodal values of $\hat u_d$ and $u$, respectively. Hence the error $w=u_d-u_c$ satisfies
\[
\|w\|_\infty\to0
\qquad\text{and}\qquad
\|w\|_{\ell_h^2}\le2\sqrt\pi\,\|w\|_\infty\to0 .
\]
By Theorem \ref{thm:discrete_minimizer}, $u_d$ is a global minimizer of $E_N$, so it solves \eqref{eq:EL_compact}, and $E_N[u_d]\le E_N[0]\le C_0$ gives $u_d\in\KN$. Hence the error equation \eqref{eq:error_equation} and Lemma \ref{lem:potential_estimate} apply, and neither of them uses \eqref{eq:discrete_stability_smallness}. Taking norms in \eqref{eq:error_equation} as in \eqref{eq:laplacian_error_bound} and using Theorem \ref{thm:EL_consistency}, we get
\[
\|\Delta_h w\|_{\ell_h^2}\le L_F\|w\|_{\ell_h^2}+\|r\|_{\ell_h^2}\le L_F\|w\|_{\ell_h^2}+5000\bigl(C_V^{3/2}+C_V\bigr)h\to0 .
\]
Finally, periodic summation by parts in each layer, followed by the Cauchy--Schwarz inequality, gives
\[
\|\nabla_h w\|_{\ell_h^2}^2=-\sum_{i=1}^2h_i\sum_{n\in I_i}w^i_n\,\Delta_iw^i_n\le\|w\|_{\ell_h^2}\,\|\Delta_h w\|_{\ell_h^2}\to0 .
\]
Therefore $\|w\|_{H_h^2}\to0$.
\end{proof}

\begin{remark}
Theorem \ref{thm:gamma} gives no rate of convergence. In the proof of Theorem \ref{thm:discrete_H2_stability}, the smallness condition \eqref{eq:discrete_stability_smallness} is used only to absorb the term $C_PL_F\|w\|_{\ell_h^2}$, and this absorption yields the rate $O(h)$ in Theorem \ref{thm:discrete_to_continuum_convergence}. The factor $\sqrt h$ in Lemma \ref{lem:energy_consistency} is a consistency bound for the energies, not a rate for the convergence of minimizers.
\end{remark}

\subsection{Numerical illustration}\label{sec:numerical}
We illustrate the preceding analysis for a Lennard--Jones interaction, examining the continuum profile and the atomistic-to-continuum errors. For horizontal separation $s$ and fixed layer separation $d>0$, the pair potential is
$$
 V(s)=4\epsilon\left[
 \frac{\sigma^{12}}{(d^2+s^2)^6}
 -\frac{\sigma^6}{(d^2+s^2)^3}\right].
$$
Here $\sqrt{d^2+s^2}$ is the distance between the interacting atoms, $\sigma$ is the Lennard--Jones length parameter, and $\epsilon$ sets the interaction strength. As in \eqref{eq:Phi_periodization}, the continuum misfit potential in energy \eqref{eq:E12} and the potential in energy \eqref{eq:v} are, respectively,
\[
 \Phi(s)=2\sum_{m\in\mathbb Z}V(s-2\pi m),\qquad W(s)=2\Phi(s).
\]
We use $(d,\sigma,\epsilon)=(2\pi,1.8\pi,16)$ in all three figures and keep these parameters fixed as $N$ increases. The continuum references in Figures \ref{fig:lj-pair-convergence}
and \ref{fig:lj-relative-convergence} are constructed from the profile in Figure \ref{fig:lj-continuum-minimizer}, using the phase alignment described below. Using the grids and notation of Section \ref{sec:discrete_min}, we compute approximate stationary configurations $u_d$ of the truncated atomistic energy, subject to $(u_d^1)+(u_d^2)=0$. We use
$$
N\in\{320,640,1280,2560, 3840,5120,7680,10240\}.
$$
In the atomistic energy, we truncate the infinite interlayer interaction sums to the indices $m=-320,\ldots,320$. Sensitivity to the interaction cutoff was checked by increasing it to $640$.

In this subsection, $u_c$ denotes a sampled continuum reference phase-aligned using the discrete layer means, unlike the fixed-phase reference in the preceding analysis. To specify the phase alignment, set
$$
q(x)=u_1(x)-u_2(x)=v_0(x)-x,\quad \delta_N=(u_d^2)-(u_d^1),
$$
and define $q_{\delta_N}(x)=q(x-\delta_N)-\delta_N$. $u_c$ is obtained by sampling $(q_{\delta_N}/2,-q_{\delta_N}/2)$ on the respective native grids. The errors in the displacement pair are
$$
 \|u_d-u_c\|_{H_h^2},\quad \|u_d-u_c\|_\infty,
$$
with the discrete norms defined in Section \ref{sec:discrete_min}. Figure \ref{fig:lj-pair-convergence} displays these two quantities against the number of layer-1 atoms.
\begin{figure}[ht]
  \centering
  \includegraphics[width=\textwidth]{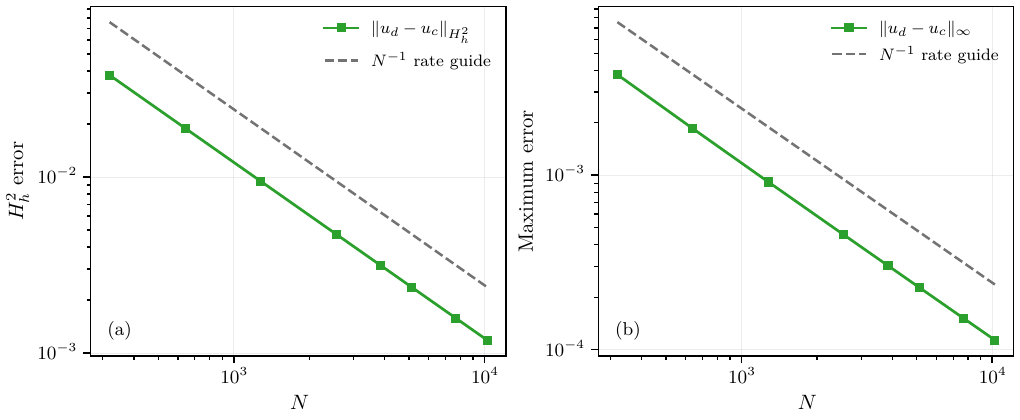}
  \caption{Errors between the computed atomistic pair $u_d$ and the phase-aligned sampled continuum reference $u_c$. Panels (a) and (b) show $\|u_d-u_c\|_{H_h^2}$ and $\|u_d-u_c\|_\infty$, respectively, against the number $N$ of layer-1 atoms. Dashed gray lines are rate guides proportional to $N^{-1}$.}
  \label{fig:lj-pair-convergence}
\end{figure}

Because the two layers are sampled on different grids, we reconstruct both displacements as functions of a common coordinate. Let $\mathcal T_N$ denote interpolation of displacement values on the uniform $N$-point grid by a real $2\pi$-periodic trigonometric polynomial of degree at most $\lfloor N/2\rfloor$. When $N$ is even, we use a cosine-only Nyquist term. We define the reconstructed relative-displacement error by
$$
e_N^{\mathrm{rel}}(x)
=(\mathcal T_Nu_d^1)(x)
-(\mathcal T_{N+1}u_d^2)(x)-q_{\delta_N}(x).
$$
Here ``relative displacement'' means the difference between the two layers, not division by a reference norm. We retain the complete reconstructed error, including all Fourier modes of these interpolants; no smoothing or filtering is applied. Figure \ref{fig:lj-relative-convergence} shows
$\|e_N^{\mathrm{rel}}\|_{H^2}$ and $\|e_N^{\mathrm{rel}}\|_{L^\infty}$, where
$$
\|e\|_{H^2}^2 =\int_{-\pi}^{\pi} \bigl(|e(x)|^2+|e'(x)|^2+|e''(x)|^2\bigr)\,dx,
\quad
\|e\|_{L^\infty}=\max_{x\in[-\pi,\pi]}|e(x)|.
$$

\begin{figure}[ht]
  \centering
  \includegraphics[width=\textwidth]{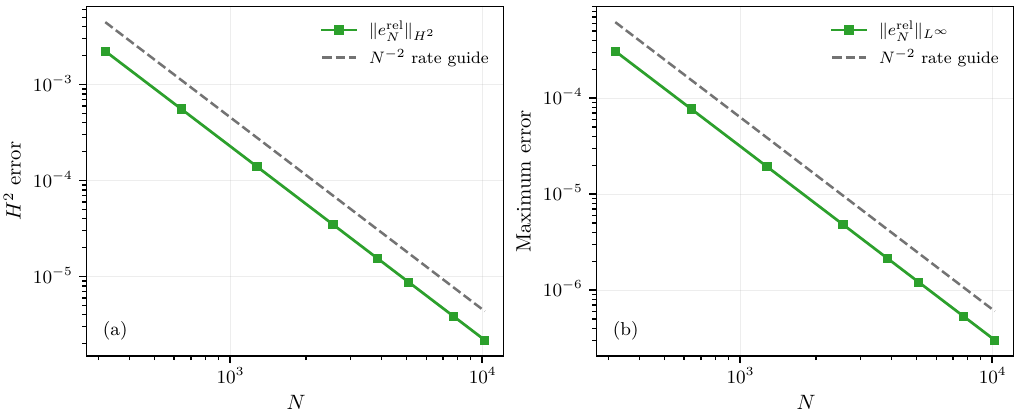}
  \caption{Errors in the reconstructed relative displacement for the same atomistic configurations and interaction parameters as Figure \ref{fig:lj-pair-convergence}. Panels (a) and (b) show $\|e_N^{\mathrm{rel}}\|_{H^2}$ and $\|e_N^{\mathrm{rel}}\|_{L^\infty}$, respectively, using the phase-aligned continuum reference. Dashed gray lines are rate guides proportional to $N^{-2}$.}
  \label{fig:lj-relative-convergence}
\end{figure}

Least-squares fits of the logarithmic errors against $\log N$ on the four preselected sizes $N=1280,2560,5120,10240$ give slopes $(-0.9999,-1.0043)$ for the pair errors and $(-1.9997,-1.9997)$ for the relative-displacement errors, in the panel order shown. These computations provide numerical evidence for first-order convergence of the computed atomistic displacement $u_d$ to the
phase-aligned continuum reference $u_c$, in both the discrete $H_h^2$ and maximum norms. For the reconstructed relative displacement, a second-order rate is observed in the continuous $H^2$ and $L^\infty$ norms. The dashed guides, proportional to $N^{-1}$ and $N^{-2}$, respectively, illustrate these rates and are vertically shifted for clarity. 
\begin{remark}
These observations concern the computed configurations. Global atomistic minimality and the smallness condition of Theorem \ref{thm:discrete_to_continuum_convergence} have not been verified for these parameters. The observed second-order rate for the relative displacement is an empirical finding, not a conclusion of the convergence theorems.
\end{remark}

\medskip
\noindent\textbf{Code and data availability.}
The source code and numerical data for
Figures~\ref{fig:lj-continuum-minimizer}--\ref{fig:lj-relative-convergence}
are available at
\url{https://github.com/HaoLan0512/1D_relaxation_lj_convergence}.
The repository includes saved continuum and atomistic profiles, allowing the figures and convergence tables to be regenerated without repeating the optimization. Instructions for installation, verification, and fresh calculations, together with details of the numerical methods and validation procedures, are provided in the accompanying documentation.

\appendix

\section{Technical details and background} \label{app:technical}

\subsection{Derivation of the simplified discrete model} \label{sec:simplified}
We now rewrite the energy \eqref{eq:energy_raw} in a simplified, rescaled form \eqref{eq:energy_discrete}. Start from intralayer terms, reindex the argument by $m'=m-n$ and rescale the displacement field as $\xi_n^{\hi}=\frac{h_{\hi}}{2 \pi} u_n^{\hi}$. Assume there exists a rescaled 1D even potential $\tilde V:\R\to\R$ such that for both layers   
\[
2 h_{\hi}^2 V_{\hi, \hi}(h_{\hi} s) = \tilde V(s), \quad \tilde V(-s) = \tilde V(s),
\]
we get
$$
E^{\hi}_{\mathrm{intra}}[u]=
\frac{1}{2 h_{\hi}}\sum_{n=0}^{N_{\hi}-1}\sum_{m' \in Z} 
\tilde V\!\left(\frac{u_n^{\hi}-u_{n+m'}^{\hi}}{2\pi}-m'\right)
$$
Now assume the unperturbed state is a nondegenerate minimum, which gives $\tilde V'\!\left(-m'\right)=0$. Expand around this minimum and impose the nearest-neighbor truncation, up to an additive constant and a scaling factor, we obtain
$$
E^{\hi}_{\mathrm{intra}}[u]=\frac{h_{\hi}}{2}
\sum_n \left(\frac{u_{n+1}^{\hi}-u_n^{\hi}}{h_{\hi}}\right)^2.
$$
For the potential terms, we first assume the pair potentials are symmetric and even, i.e. $V^{12} = V^{21} = \bar V$. Reindex the argument by $m'=m-n$ and rescale the displacement field as $\xi_n^{\hi}=\frac{h_{\hi}}{2 \pi} u_n^{\hi}$, we obtain for ${\hi},{\hj}\in\{1,2\}, \,{\hi}\neq {\hj}$
$$
E^{\hi}_{\mathrm{inter}} = h_{\hi} \sum_{n=0}^{N_{\hi}-1}\sum_{m' \in \Z}\bar V\left((h_{\hi}-h_{\hj})n+\frac{h_{\hi}}{2\pi}u^{\hi}_n-\frac{h_{\hj}}{2\pi}u^{\hj}_{n+m'}-h_{\hj} m'\right),
$$
Rescale the interlayer potential as $\bar V(x) = V (\frac{4\pi}{h_1+h_2}x)$ and denote the harmonic mean $h = \frac{2 h_1 h_2}{h_1 + h_2}$, we obtain
$$
E^{\hi}_{\mathrm{inter}} = h_{\hi} \sum_{n=0}^{N_{\hi}-1}\sum_{m' \in \Z}V\left((-1)^{{\hj}}hn + \frac{h}{h_{\hj}}u^{\hi}_n - \frac{h}{h_{\hi}}u^{\hj}_{n+m'} - 2\pi \frac{h}{h_{\hi}} m'\right).
$$

\subsection{Proof of Corollary \ref{lem:mean_fixing}}
Choose $c=c(d)$ so that the transformed means sum to zero:
$$
\left(\tilde{u}^1\right)+\left(\tilde{u}^2\right)=0 .
$$
Since
$$
\left(\tilde{u}^1\right)+\left(\tilde{u}^2\right)=\left(u^1\right)+\left(u^2\right)+h_2 d+\left(h_1+h_2\right) c,
$$
take
$$
c(d):=-\frac{\left(u^1\right)+\left(u^2\right)+h_2 d}{h_1+h_2} .
$$
Then indeed
$$
\left(\tilde{u}^1\right)+\left(\tilde{u}^2\right)=0 .
$$
Substitute this into $\left(\tilde{u}^1\right)$ :
$$
\left(\tilde{u}^1\right)(d)=\left(u^1\right)+h_2 d-h_2 \frac{\left(u^1\right)+\left(u^2\right)+h_2 d}{h_1+h_2}=\frac{h_1\left(u^1\right)-h_2\left(u^2\right)+h_1 h_2 d}{h_1+h_2} .
$$
Choose $\ell_*$ to be the nearest integer to
$$
\frac{h_2\left(u^2\right)-h_1\left(u^1\right)}{h_1 h_2} .
$$
Then
$$
\left|\frac{h_2\left(u^2\right)-h_1\left(u^1\right)}{h_1 h_2}-\ell_*\right| \leq \frac{1}{2},
$$
which gives
$$
\left|\left(\tilde{u}^1\right)\left(\ell_*\right)\right| \leq \frac{h_1 h_2}{2\left(h_1+h_2\right)} .
$$
Since $\left(\tilde{u}^2\right)=-\left(\tilde{u}^1\right)$, the same bound holds for layer 2 . Finally,
$$
\frac{h_1 h_2}{2\left(h_1+h_2\right)}=\frac{1}{4} \frac{2 h_1 h_2}{h_1+h_2}=\frac{h}{4} .
$$

\subsection{Proof of Lemma \ref{lem:nf_split}} \label{appendix:nf_split}
Note that by \eqref{eq:argument}, we have the estimate
\[
    |\Anm[u]| \geq \pi |m|-2\pi(1+\left\|u^1\right\|_{\infty, 1}+\left\|u^2\right\|_{\infty, 2} ).
\]
We split the sum into a near field and a far field. Let
\[
    M:=\left\lceil 4(1+\left\|u^1\right\|_{\infty, 1}+\left\|u^2\right\|_{\infty, 2} )\right\rceil .
\]
For \(|m|\le M\), we only use the bound \(|G|\le C_V\), then
\[
    \sum_{|m|\le M}|G(\Anm[u])|
    \le
    C_V(2M+1) \leq 11C_V \left(1+\left\|u^1\right\|_{\infty, 1}+\left\|u^2\right\|_{\infty, 2}\right) 
\]
For \(|m|>M\), we have
\[
    \pi|m|-2\pi(1+\left\|u^1\right\|_{\infty, 1}+\left\|u^2\right\|_{\infty, 2})
    \ge
    \frac{\pi}{2}|m|,
\]
and hence for any $n\in I_{\hi}$, since $r \geq 2$,
\[
    \sum_{|m|>M}|G(\Anm[u])|
    \le
    C_V
    \sum_{|m|>M}
    \left(
        1+\frac{\pi}{2}|m|
    \right)^{-r} 
    \le
    2C_V\sum_{m=1}^{\infty}
    \left(
        1+\frac{\pi}{2}m
    \right)^{-r}
    \le 2C_V.
\]
Combining the near-field and far-field estimates gives
\[
    \sum_{m\in\mathbb Z}\sup_{n\in I_{\hi}}|G(\Anm[u])| \leq 13C_V \left(1+\left\|u^1\right\|_{\infty, 1}+\left\|u^2\right\|_{\infty, 2}\right).
\]

\subsection{Proof of Lemma \ref{lem:uniform_boundedness}} \label{appendix:lem3_proof}
Following the Sobolev--Morrey embedding as in \cite[\S5.6.2, Thm. 5]{evans2010partial} and applying the 1D finite-difference analogue as in \cite[\S5.10, Ex. 4]{evans2010partial}, we have for any indices $j, k$, 
$$
\left|u^{\hi}_k - u^{\hi}_j\right| \leq \sqrt{2 \pi} \left\|\nabla_{\hi}^{+}u^{\hi}\right\|_{\ell^2_{\hi}} .
$$
Then, since the layer mean lies between the minimum and maximum values,
$$
\left\|u^{\hi}-\left(u^{\hi}\right)\right\|_{\infty, \hi} \leq \sqrt{2 \pi}\left\|\nabla_{\hi}^{+} u^{\hi}\right\|_{\ell_{\hi}^2} .
$$
With the mean-fixing normalization in Corollary \ref{lem:mean_fixing}, we have
$$
\left\|u^1\right\|_{\infty, 1}+\left\|u^2\right\|_{\infty, 2} \leq 2 \sqrt{\pi} \sqrt{G[u]}+\frac{h}{2},
$$
where
$$
G[u]:=\left\|\nabla_{h_1}^{+} u^1\right\|_{\ell_1^2}^2+\left\|\nabla_{h_2}^{+} u^2\right\|_{\ell_2^2}^2 .
$$
Then Lemma \ref{lem:nf_split} gives the interaction lower bound
$$
\sum_{i=1}^2 h_i \sum_{n \in I_{\hi}} \sum_{m \in \Z} V\left(A_{n, m}^{\hi}[u]\right) \geq
-52 \pi C_V(1+\left\|u^1\right\|_{\infty, 1}+\left\|u^2\right\|_{\infty, 2}) .
$$
hence
$$
E_N[u] \geq \frac{1}{2} G[u]-52 \pi C_V\left(1+\frac{h}{2}+2 \sqrt{\pi} \sqrt{G[u]}\right)
$$
Since $E_N[u] \leq C_0 \leq 52 \pi C_V$, it follows that $G[u]$ is uniformly bounded, and one can calculate that the sum of the layerwise supremum norms is bounded by
$$
\left\|u^1\right\|_{\infty, 1}+\left\|u^2\right\|_{\infty, 2} \leq 
\frac{h}{2} + 42\pi \sqrt{C_V} + 416\pi^2 C_V.
$$

\subsection{Proof of Lemma \ref{lem:EL_lipschitz}} \label{app:EL_lipschitz}
By Lemma \ref{lem:nf_split}, applied with \(V'\) in place of \(G\), the
discrete maximum-norm bound on \(u\) implies
\[
\sum_{m\in\mathbb Z}\sup_{n\in I_{\hi}}
\left|V'\left(A^{\hi}_{n,m}[u]\right)\right|
\leq 13C_V \left(1+\left\|u^1\right\|_{\infty, 1}+\left\|u^2\right\|_{\infty, 2}\right).
\]
Therefore, using the Euler--Lagrange equation \eqref{eq:EL_compact},
\[
\|\Delta_{{\hi}}u^{\hi}\|_{\infty, \hi}
=
\|F^{\hi}(u)\|_{\infty, \hi}
\leq 26C_V \left(1+\left\|u^1\right\|_{\infty, 1}+\left\|u^2\right\|_{\infty, 2}\right).
\]
Then w.l.o.g. for any indices \(n,m\) within a period, we have
\[
|\nabla_{\hi}^+u^{\hi}_n-\nabla_{\hi}^+u^{\hi}_m|
\le h_{\hi}|n-m|
\|\Delta_{{\hi}}u^{\hi}\|_{\infty, \hi}
\leq 2\pi \|\Delta_{{\hi}}u^{\hi}\|_{\infty, \hi}.
\]
Moreover, periodicity gives
\[
\sum_{n\in I_{\hi}} \nabla_{\hi}^+u^{\hi}_n
=
\sum_{n\in I_{\hi}}\frac{u^{\hi}_{n+1}-u^{\hi}_n}{h_{\hi}}=0.
\]
Thus, unless every $\nabla_{\hi}^+u^{\hi}_n=0$, for each $n$ there exists some $m$ such that $\nabla_{\hi}^+u^{\hi}_n$ and $\nabla_{\hi}^+u^{\hi}_m$ have opposite signs. Consequently, 
$$
\left|\nabla_{\hi}^+u^{\hi}_n\right| \leq \left|\nabla_{\hi}^+u^{\hi}_n-\nabla_{\hi}^+u^{\hi}_m\right| \leq 2\pi \|\Delta_{{\hi}}u^{\hi}\|_{\infty, \hi}.
$$
Therefore,
$$
\left\|\nabla_i^{+} u^i\right\|_{\infty, i} \leq 2 \pi\left\|\Delta_i u^i\right\|_{\infty, i} \leq 52 \pi C_V\left(1+\left\|u^1\right\|_{\infty, 1}+\left\|u^2\right\|_{\infty, 2}\right) .
$$
Finally, by telescoping and applying Lemma \ref{lem:uniform_boundedness}, we have
$$
\begin{aligned}
|u^{\hi}_{n+m}-u^{\hi}_n|
& \le
\sum_{r=0}^{|m|-1} h_{\hi}
\|\nabla_{\hi}^+u^{\hi}\|_{\infty, \hi}
\leq 52\pi C_V (1+\left\|u^1\right\|_{\infty, 1}+\left\|u^2\right\|_{\infty, 2}) |m| h_{\hi} \\
& \leq 3.3 \times 10^4 \pi^3\left(C_V+C_V^2\right) |m| h.
\end{aligned}
$$

\subsection{Proof of Lemma \ref{lem:approx_symmetry}} \label{app:approx_symmetry}
For simplicity, denote only in this proof
$$
\Lambda_V:=3.3 \times 10^4 \pi^3\left(C_V+C_V^2\right),
$$
$$
U:= \left\|u^1\right\|_{\infty, 1}+\left\|u^2\right\|_{\infty, 2}.
$$
Recall that $A^{\hi}_{n,m}\left[u\right]$ is defined as
\begin{equation*}
    \Anm \left[u\right] = (-1)^{{\hj}} h n + \frac{h}{h_{\hj}} u_{n}^{\hi} - \frac{h}{h_{\hi}} u_{n+m}^{\hj} - 2 \pi \frac{h}{h_{\hi}} m , \quad n \in I_{\hi}.
\end{equation*}
Now for \eqref{eq:energy_localized}, define
\begin{equation*}
    \Bnm \left[u\right] := (-1)^{\hj}h n+u_n^{\hi}-u_n^{\hj}-2 \pi m, \quad n \in I_{\hi}.
\end{equation*}
Then
$$
E^{{\hi}}\left[u\right] - \Eloc^{{\hi}}\left[u\right] = \avsum_{n \in I_{\hi}}\sum_{m \in \Z}\left(V(\Anm)-V(\Bnm)\right).
$$
With $\frac{h_1-h_2}{h_1+h_2} = \frac{h}{4\pi}$, one can estimate
$$
\left|\Anm-\Bnm\right| \leq \frac{h}{4\pi}U
+\left(1+\frac{h}{4\pi}\right)\left|u_{n+m}^{\hj}-u_n^{\hj}\right|+ \frac{h}{2} \left|m\right| .
$$
Note that by Lemma \ref{lem:EL_lipschitz}, for all $n \in I_{\hi}$, $m \in \Z$,
$$
\left|u_{n+m}^{\hi}-u_n^{\hi}\right| \leq \Lambda_V |m| h,
$$  
Then 
\begin{equation} \label{eq:AB_diff}
\left|\Anm-\Bnm\right| \leq h\left[\frac{U}{4 \pi}+\left(\frac{4}{3} \Lambda_V+\frac{1}{2}\right)|m|\right] .
\end{equation}
Fix $n$, by mean value theorem, 
\begin{equation} \label{eq:V_est}
\left|V\left(\Anm\right)-V\left(\Bnm\right)\right| \leq \sup _{0 \leq \theta \leq 1}\left|V'\left(C_{n, m, \theta}^i\right)\right|\left|\Anm-\Bnm\right|,
\end{equation}
where
$$
C_{n, m, \theta}^i:=B_{n, m}^i+\theta\left(A_{n, m}^i-B_{n, m}^i\right), \quad 0 \leq \theta \leq 1.
$$
The same lower-bound argument as in Lemma \ref{lem:nf_split} gives
$$
\left|C_{n, m, \theta}^i\right| \geq \pi|m|-2 \pi(1+U),
$$
uniformly in
$$
i=1,2, \quad n \in I_i, \quad m \in \mathbb{Z}, \quad \theta \in[0,1] .
$$
Choose $M=\left\lceil 4(1+U)\right\rceil$. Now apply Lemma \ref{lem:nf_split} to $V'$, we retain
\begin{equation} \label{eq:Vprime_sum}
\sum_{m \in \mathbb{Z}} \sup _{\substack{n \in I_i \\ 0 \leq \theta \leq 1}}\left|V'\left(C_{n, m, \theta}^i\right)\right| \leq 13 C_V (1+U) .
\end{equation}
For the $|m|$-weighted sum, we retain
\begin{equation} \label{eq:Vprime_msum}
\sum_{m \in \mathbb{Z}}|m| \sup _{\substack{n \in I_i \\ 0 \leq \theta \leq 1}}\left|V'\left(C_{n, m, \theta}^i\right)\right| \leq 31 C_V(1+U)^2 .
\end{equation}
Now we can combine \eqref{eq:AB_diff}, \eqref{eq:V_est}, \eqref{eq:Vprime_sum} and \eqref{eq:Vprime_msum} to obtain
$$
\left|E^i[u]-E_{\text {loc }}^i[u]\right| \leq 2 \pi h\left[\frac{U}{4 \pi}\left(13 C_V (1+U)\right)+\left(\frac{4}{3} \Lambda_V+\frac{1}{2}\right)\left(31 C_V (1+U)^2\right)\right].
$$
To simplify the above estimate, use $U(1+U) \leq (1+U)^2$, plug in the definition of $\Lambda_V$ and relax the constant terms, we obtain
$$
\left|E^i[u]-E_{\mathrm{loc}}^i[u]\right| \leq h C_V(1+U)^2\left[34 \pi+2.8 \times 10^6 \pi^4\left(C_V+C_V^2\right)\right].
$$
Next, Lemma \ref{lem:uniform_boundedness} gives
$$
U \leq \frac{h}{2}+42 \pi \sqrt{C_V}+416 \pi^2 C_V
$$
Substituting this into the previous inequality, we get
$$
\left|E^i[u]-E_{\mathrm{loc}}^i[u]\right| \leq h \mathcal{P}\left(C_V\right).
$$

Now fix $k=(k_1,k_2)\in\Z^2$ and set
\[
U_k:=
\|(\tau_k u)^1\|_{\infty,1}
+
\|(\tau_k u)^2\|_{\infty,2}.
\]
By \eqref{symmetry_localized},
\[
U_k\le U+2\pi(|k_1|+|k_2|).
\]
Moreover, $\tau_k$ does not change discrete differences, so
\[
|(\tau_k u)^i_{n+m}-(\tau_k u)^i_n|
=
|u^i_{n+m}-u^i_n|
\le \Lambda_V |m|h.
\]
Therefore, repeating the estimates above with $U$ replaced by $U_k$ gives, for
each fixed $k$,
\[
\left|E^i[\tau_k u]-E^i_{\rm loc}[\tau_k u]\right|
\le C(k,C_V)h,
\]
where $C(k,C_V)$ is independent of $N$. Since
$E^i_{\rm loc}[\tau_k u]=E^i_{\rm loc}[u]$, the triangle inequality gives
\[
\begin{aligned}
\left|E^i[\tau_k u]-E^i[u]\right|
&\le
\left|E^i[\tau_k u]-E^i_{\rm loc}[\tau_k u]\right|
+
\left|E^i_{\rm loc}[u]-E^i[u]\right|  \\
&\le C(k,C_V)h,
\end{aligned}
\]
enlarging $C(k,C_V)$ if necessary.

\subsection{Proof of Lemma \ref{lem:potential_estimate}} \label{app:potential_estimate}
Recall the uniform bound for $u_c$
$$
\left\|u_c^1\right\|_{\infty, 1}+\left\|u_c^2\right\|_{\infty, 2} \leq 2 \pi .
$$
and for $u_d$
$$
\left\|u_d^1\right\|_{\infty, 1}+\left\|u_d^2\right\|_{\infty, 2} \leq 
42\pi \sqrt{C_V} + 416\pi^2 C_V + \frac{2\pi}{3}.
$$
Then for any $t \in [0,1]$, 
$$
\begin{aligned}
\left\|z_t^1\right\|_{\infty, 1}+\left\|z_t^2\right\|_{\infty, 2} &\leq (1-t)\left(\left\|u_c^1\right\|_{\infty, 1}+\left\|u_c^2\right\|_{\infty, 2}\right) 
+t\left(\left\|u_d^1\right\|_{\infty, 1}+\left\|u_d^2\right\|_{\infty, 2}\right) \\
& \leq \frac{8\pi}{3}+42 \pi \sqrt{C_V}+416 \pi^2 C_V.
\end{aligned}
$$
We next estimate the \(V''\)-sum uniformly in \(t\). The lower-bound calculation used in the proof of Lemma \ref{lem:nf_split} gives
\[
\left|A_{n,m}^i[z_t]\right| \geq \pi|m|-2\pi\left(1+\|z_t^1\|_{\ell_1^\infty}+\|z_t^2\|_{\ell_2^\infty}\right).
\]
Denote only in this proof 
\[
\mathcal M_V:=\max_{i=1,2}\sup_{N \geq 2}\sup_{n\in I_i}\sum_{m\in\mathbb Z}\sup_{t\in[0,1]}\left|V''\left(A_{n,m}^i[z_t]\right)\right|,
\]
then by the same argument as in Lemma \ref{lem:nf_split}, we have
$$
\mathcal M_V \leq 13 C_V \left(1+\frac{8\pi}{3}+42 \pi \sqrt{C_V}+416 \pi^2 C_V\right)
$$
It remains to prove the force estimate. For each fixed $n$, apply the mean value theorem on \eqref{eq:force_map}, we have  
$$
\left|F^i_n(u_d)-F^i_n(u_c)\right| \leq 2 \sum_{m \in \Z} \sup_{t \in [0,1]} \left|V''\left(A_{n,m}^i\left[z_{t}\right]\right)\right|\left(\frac{h}{h_{\hj}}\left|w^i_n\right|+\frac{h}{h_{\hi}}\left|w^{\hj}_{n+m}\right|\right)
$$
For the term with $w_n^i$, we can directly use the assumption \eqref{Assumption} to bound it by $2 \frac{h}{h_{\hj}} \mathcal M_V \left|w_n^i\right|$. Square the above inequality, and use Cauchy--Schwarz for the $w_{n+m}^{\hj}$ term, we have
$$
\begin{aligned}
& \left|F^i_n(u_d)-F^i_n(u_c)\right|^2  \\
& \leq  8 \left(\frac{h}{h_{\hj}}\right)^2 \mathcal M_V^2 \left|w_n^i\right|^2
 + 8 \left(\frac{h}{h_{\hi}}\right)^2 \mathcal M_V \left(\sum_{m \in \Z}\sup_{t \in [0,1]}\left|V''\left(A_{n,m}^i\left[z_{t}\right]\right)\right| \left|w^{\hj}_{n+m}\right|^2\right).
\end{aligned}
$$
Multiply both sides by $h_{\hi}$ and sum over $n \in I_{\hi}$, we have
$$
\begin{aligned}
& \left\|F^i(u_d)-F^i(u_c)\right\|_{\ell^2_i}^2 \leq  8 \left(\frac{h}{h_{\hj}}\right)^2 \mathcal M_V^2 \left\|w^i\right\|_{\ell^2_i}^2
 \\
& + 8 \left(\frac{h}{h_{\hi}}\right)^2 \mathcal M_V h_{\hi}  \sum_{n \in I_{\hi}} \sum_{m \in \Z}\sup_{t \in [0,1]}\left|V''\left(A_{n,m}^i\left[z_{t}\right]\right)\right| \left|w^{\hj}_{n+m}\right|^2.
\end{aligned}
$$
Note that for fixed $n \in I_{\hi}$, $p \in I_{\hj}$ and for any $q \in \Z$, the terms $\Anm$ are related by
$$
A^{\hi}_{n,p-n+q N_{\hj}} = -A^{\hj}_{p,n-p-q N_{\hi}}.
$$
Take $m = p-n+q N_{\hj}$ for some $q \in \Z$. For fixed $n$, the mapping 
$$
(p, q) \in I_{\hj} \times \Z \longmapsto m:=p-n+q N_{\hj} \in \Z
$$
is a bijection. Therefore, we can rewrite the second term as
$$
8 \left(\frac{h}{h_{\hi}}\right)^2 \mathcal M_V h_{\hi}  \sum_{n \in I_{\hi}} \sum_{p \in I_{\hj}} \sum_{q \in \Z}\sup_{t \in [0,1]}\left|V''\left(A_{n,p-n+q N_{\hj}}^i\left[z_{t}\right]\right)\right| \left|w^{\hj}_{p}\right|^2.
$$
Switch the order of summation, use the above relation between $\Anm$ and the evenness of $V''$ gives
$$
8 \left(\frac{h}{h_{\hi}}\right)^2 \mathcal M_V h_{\hi}  \sum_{p \in I_{\hj}} \sum_{n \in I_{\hi}} \sum_{q \in \Z} \sup_{t \in [0,1]}\left|V''\left(A_{p,n-p-q N_{\hi}}^{\hj}\left[z_{t}\right]\right)\right| \left|w^{\hj}_{p}\right|^2.
$$
Now for each fixed $p \in I_{\hj}$, the mapping
$$
(n, q) \in I_{\hi} \times \Z \longmapsto m':=n-p-q N_{\hi} \in \Z
$$
is a bijection, thus we can rewrite the second term as
$$
8 \left(\frac{h}{h_{\hi}}\right)^2 \mathcal M_V h_{\hi}  \sum_{p \in I_{\hj}} \sum_{m' \in \Z} \sup_{t \in [0,1]}\left|V''\left(A_{p,m'}^{\hj}\left[z_{t}\right]\right)\right| \left|w^{\hj}_{p}\right|^2.
$$
Now we can use the assumption \eqref{Assumption} to bound it by
$$
8 \left(\frac{h}{h_{\hi}}\right)^2 \left(\frac{h_{\hi}}{h_{\hj}}\right) \mathcal M_V^2 \left\|w^{\hj}\right\|_{\ell^2_{\hj}}^2.
$$
Combine the two terms, we have
$$
\left\|F^i(u_d)-F^i(u_c)\right\|_{\ell^2_i}^2 \leq 8 \left(\frac{h}{h_{\hj}}\right)^2 \mathcal M_V^2 \left\|w^i\right\|_{\ell^2_i}^2 + 8 \left(\frac{h}{h_{\hi}}\right)^2 \left(\frac{h_{\hi}}{h_{\hj}}\right) \mathcal M_V^2 \left\|w^{\hj}\right\|_{\ell^2_{\hj}}^2.
$$
Combine the two layers estimates,
$$
\left\|F\left(u_d\right)-F\left(u_c\right)\right\|_{\ell_h^2}^2 \leq 32 \mathcal{M}_V^2\left(\left\|w^1\right\|_{\ell^2_1}^2+\left\|w^2\right\|_{\ell^2_2}^2\right) .
$$
Therefore,
$$
\left\|F\left(u_d\right)-F\left(u_c\right)\right\|_{\ell_h^2} \leq 4 \sqrt{2} \mathcal{M}_V\left\|u_d-u_c\right\|_{\ell_h^2}.
$$
Set \(L_F:=4\sqrt2\,\mathcal M_V\). This constant is bounded above by a function depending only on \(C_V\). In particular, we have the upper estimate
$$
L_F \leq 2.2 \times 10^4 \sqrt{2} \pi^2\left(C_V+C_V^2\right) .
$$


\subsection{Background on \texorpdfstring{$\Gamma$}{Gamma}-convergence}\label{appendix:gamma}
We recall the notions used in Section \ref{sec:gamma}; see \cite{10.1093/acprof:oso/9780198507840.001.0001,BraidesGelli2006} for more details. Let $(X,d)$ be a metric space, and let $F_N,F\colon X\to(-\infty,+\infty]$ for $N\ge2$.

\begin{definition}\label{def:gamma}
We say that $F_N$ \emph{$\Gamma$-converges} to $F$ in $X$ if the following two conditions hold.
\begin{enumerate}
\item[(i)] \emph{Liminf inequality:} for every $x\in X$ and every sequence $x_N\to x$ in $X$,
\[
F(x)\le\liminf_{N\to\infty}F_N(x_N).
\]
\item[(ii)] \emph{Recovery sequence:} for every $x\in X$ there exists a sequence $x_N\to x$ in $X$ such that
\[
\limsup_{N\to\infty}F_N(x_N)\le F(x).
\]
\end{enumerate}
The sequence $(F_N)$ is \emph{equicoercive} if for every $c\in\R$ there is a compact set $K_c\subset X$ such that $\{x\in X:F_N(x)\le c\}\subset K_c$ for all $N$.
\end{definition}

The importance of $\Gamma$-convergence for minimization problems comes from the following theorem.
\begin{theorem*}[Fundamental theorem of $\Gamma$-convergence {\cite[Ch.~1]{10.1093/acprof:oso/9780198507840.001.0001}}]
Assume that $F_N$ $\Gamma$-converges to $F$, that $(F_N)$ is equicoercive, and that $F\not\equiv+\infty$. Then $F$ attains its minimum on $X$, and
\[
\min_XF=\lim_{N\to\infty}\inf_XF_N.
\]
Moreover, if $F_N(x_N)-\inf_XF_N\to0$, then $(x_N)$ is precompact in $X$, and each of its limit points is a minimizer of $F$.
\end{theorem*}

In Section \ref{sec:gamma}, $X=L^2(-\pi,\pi)^2$, and each discrete pair $u$ is identified with its piecewise affine interpolant $\hat u$.

\bibliographystyle{amsplain} 
\bibliography{references}

@book{evans2010partial,
  title={Partial Differential Equations},
  author = {Evans, Lawrence C.},
  volume={19},
  year={2010},
  publisher={American Mathematical Society},
  address={Providence, RI},
  edition={2nd},
  isbn={978-0-8218-4974-3}
}

@book{dacorogna_direct,
author = {Dacorogna, Bernard},
address = {New York},
booktitle = {Direct methods in the calculus of variations},
edition = {2nd},
isbn = {9780387357799},
lccn = {2007938908},
publisher = {Springer},
series = {Applied mathematical sciences ; vol. 78},
title = {Direct methods in the calculus of variations },
year = {2008},
}

@article{Cao2018Superconductivity,
  author  = {Cao, Yuan and Fatemi, Valla and Fang, Shiang and Watanabe, Kenji and Taniguchi, Takashi and Kaxiras, Efthimios and Jarillo-Herrero, Pablo},
  title   = {Unconventional superconductivity in magic-angle graphene superlattices},
  journal = {Nature},
  volume  = {556},
  pages   = {43--50},
  year    = {2018},
  doi     = {10.1038/nature26160}
}

@article{Park2023FQAH,
  author  = {Park, Heonjoon and Cai, Jiaqi and Anderson, Eric and Zhang, Yinong and Zhu, Jiayi and Liu, Xiaoyu and Wang, Chong and Holtzmann, William and Hu, Chaowei and Liu, Zhaoyu and Taniguchi, Takashi and Watanabe, Kenji and Chu, Jiun-Haw and Cao, Ting and Fu, Liang and Yao, Wang and Chang, Cui-Zu and Cobden, David and Xiao, Di and Xu, Xiaodong},
  title   = {Observation of fractionally quantized anomalous Hall effect},
  journal = {Nature},
  volume  = {622},
  pages   = {74--79},
  year    = {2023},
  doi     = {10.1038/s41586-023-06536-0}
}

@article{Dai2016MoireTwist,
  author  = {Dai, Shuyang and Xiang, Yang and Srolovitz, David J.},
  title   = {Twisted Bilayer Graphene: Moir{\'e} with a Twist},
  journal = {Nano Letters},
  volume  = {16},
  number  = {9},
  pages   = {5923--5927},
  year    = {2016},
  doi     = {10.1021/acs.nanolett.6b02870}
}

@article{NamKoshino2017Relaxation,
  author  = {Nam, Nguyen N. T. and Koshino, Mikito},
  title   = {Lattice relaxation and energy band modulation in twisted bilayer graphene},
  journal = {Physical Review B},
  volume  = {96},
  pages   = {075311},
  year    = {2017},
  doi     = {10.1103/PhysRevB.96.075311},
  note    = {See also Erratum, Phys. Rev. B 101, 099901 (2020)}
}

@article{CarrEtAl2018Relaxation,
  author  = {Carr, Stephen and Massatt, Daniel and Torrisi, Steven B. and Cazeaux, Paul and Luskin, Mitchell and Kaxiras, Efthimios},
  title   = {Relaxation and domain formation in incommensurate two-dimensional heterostructures},
  journal = {Physical Review B},
  volume  = {98},
  pages   = {224102},
  year    = {2018},
  doi     = {10.1103/PhysRevB.98.224102}
}

@article{Yoo2019Reconstruction,
  author  = {Yoo, Hyobin and Engelke, Rebecca and Carr, Stephen and Fang, Shiang and Zhang, Kuan and Cazeaux, Paul and Sung, Suk Hyun and Hovden, Robert and Tsen, Adam W. and Taniguchi, Takashi and Watanabe, Kenji and Yi, Gyu-Chul and Kim, Miyoung and Luskin, Mitchell and Tadmor, Ellad B. and Kaxiras, Efthimios and Kim, Philip},
  title   = {Atomic and electronic reconstruction at the van der Waals interface in twisted bilayer graphene},
  journal = {Nature Materials},
  volume  = {18},
  number  = {5},
  pages   = {448--453},
  year    = {2019},
  doi     = {10.1038/s41563-019-0346-z}
}

@article{CazeauxLuskinMassatt2020,
  author  = {Cazeaux, Paul and Luskin, Mitchell and Massatt, Daniel},
  title   = {Energy Minimization of Two Dimensional Incommensurate Heterostructures},
  journal = {Archive for Rational Mechanics and Analysis},
  volume  = {235},
  pages   = {1289--1325},
  year    = {2020},
  doi     = {10.1007/s00205-019-01444-y}
}

@article{CazeauxEtAl2023DomainWalls,
  author  = {Cazeaux, Paul and Clark, Drake and Engelke, Rebecca and Kim, Philip and Luskin, Mitchell},
  title   = {Relaxation and Domain Wall Structure of Bilayer Moir{\'e} Systems},
  journal = {Journal of Elasticity},
  volume  = {154},
  pages   = {443--466},
  year    = {2023},
  doi     = {10.1007/s10659-023-10013-0}
}

@incollection{BraidesGelli2006,
  author    = {Braides, Andrea and Gelli, Maria Stella},
  title     = {From discrete systems to continuous variational problems: an introduction},
  booktitle = {Topics on Concentration Phenomena and Problems with Multiple Scales},
  editor    = {Braides, Andrea and Chiad{\`o} Piat, Valeria},
  series    = {Lecture Notes of the Unione Matematica Italiana},
  volume    = {2},
  pages     = {3--77},
  publisher = {Springer},
  address   = {Berlin, Heidelberg},
  year      = {2006},
  doi       = {10.1007/978-3-540-36546-4_1}
}

@article{Wang_2025,
doi = {10.1088/2632-2153/ae1807},
url = {https://doi.org/10.1088/2632-2153/ae1807},
year = {2025},
month = {nov},
publisher = {IOP Publishing},
volume = {6},
number = {4},
pages = {045040},
author = {Wang, Yangshuai and Clark, Drake and Das, Sambit and Zhu, Ziyan and Massatt, Daniel and Gavini, Vikram and Luskin, Mitchell and Ortner, Christoph},
title = {An atomic cluster expansion potential for twisted multilayer graphene},
journal = {Machine Learning: Science and Technology}
}

@misc{tu2026relaxationincommensuratestructuresquantum,
      title={Relaxation of Incommensurate Structures via Quantum Models}, 
      author={Mengfan Tu and Huajie Chen and Daniel Massatt},
      year={2026},
      eprint={2606.31104},
      archivePrefix={arXiv},
      primaryClass={physics.comp-ph},
      url={https://arxiv.org/abs/2606.31104}, 
}

@book{10.1093/acprof:oso/9780198507840.001.0001,
    author = {Braides, Andrea},
    title = {Gamma-Convergence for Beginners},
    publisher = {Oxford University Press},
    year = {2002},
    month = {07},
    isbn = {9780198507840},
    doi = {10.1093/acprof:oso/9780198507840.001.0001},
    url = {https://doi.org/10.1093/acprof:oso/9780198507840.001.0001},
}

@article{PhysRevE.96.033003,
  title = {Discrete-to-continuum modeling of weakly interacting incommensurate chains},
  author = {Espa\~nol, Malena I. and Golovaty, Dmitry and Wilber, J. Patrick},
  journal = {Phys. Rev. E},
  volume = {96},
  issue = {3},
  pages = {033003},
  numpages = {12},
  year = {2017},
  month = {Sep},
  publisher = {American Physical Society},
  doi = {10.1103/PhysRevE.96.033003},
  url = {https://link.aps.org/doi/10.1103/PhysRevE.96.033003}
}

@article{Espanol,
    author  = {Espa{\~n}ol, Malena I. and Golovaty, Dmitry and Wilber, J. Patrick},
    title = {Discrete-to-continuum modelling of weakly interacting incommensurate two-dimensional lattices},
    journal = {Proceedings of the Royal Society A: Mathematical, Physical and Engineering Sciences},
    volume = {474},
    number = {2209},
    pages = {20170612},
    year = {2018},
    month = {01},
    issn = {1364-5021},
    doi = {10.1098/rspa.2017.0612},
    url = {https://doi.org/10.1098/rspa.2017.0612},
    eprint = {https://royalsocietypublishing.org/rspa/article-pdf/doi/10.1098/rspa.2017.0612/631918/rspa.2017.0612.pdf},
}

@misc{golovaty2025hierarchyscalesmodelingweakly,
      title={On the Hierarchy of Scales in Modeling of Weakly Interacting Chains of Atoms}, 
      author={Dmitry Golovaty and J. Patrick Wilber},
      year={2025},
      eprint={2510.12945},
      archivePrefix={arXiv},
      primaryClass={math-ph},
      url={https://arxiv.org/abs/2510.12945}, 
}

@article{Cazeaux2017,
   author = {Paul Cazeaux and Mitchell Luskin and Ellad B. Tadmor},
   doi = {10.1137/16M1076198},
   issn = {15403467},
   issue = {1},
   journal = {Multiscale Modeling and Simulation},
   pages = {56-73},
   publisher = {Society for Industrial and Applied Mathematics Publications},
   title = {Analysis of rippling in incommensurate one-dimensional coupled chains},
   volume = {15},
   year = {2017},
}

@article{2016Cazeaux,
	author = {{Cazeaux, Paul} and {Luskin, Mitchell}},
	title = {Cauchy–Born strain energy density for coupled incommensurate elastic chains},
	DOI= "10.1051/m2an/2017057",
	url= "https://doi.org/10.1051/m2an/2017057",
	journal = {ESAIM: M2AN},
	year = 2018,
	volume = 52,
	number = 2,
	pages = "729-749",
}

@article{jingzhi2025formaljustificationcontinuumrelaxation,
    title={Formal justification of a continuum relaxation model for one-dimensional moir\'e materials}, 
    author={Zhou, Jingzhi (David) and Alexander B. Watson},
    year={2025},
    journal={arxiv.org/abs/2412.08854}, 
}
\end{document}